\documentclass[preprint,showpacs,preprintnumbers,amsmath,amssymb]{revtex4-2}

\usepackage{graphicx}
\usepackage[latin1]{inputenc}
\usepackage{dcolumn}
\usepackage{bm}
\usepackage{epsfig}
\usepackage{xcolor}
\usepackage{framed}
\usepackage{amsmath}
\usepackage{amssymb}
\usepackage{amsthm}
\usepackage{tikz-cd}

\graphicspath{{Images_Paper/}}
\usepackage{hyperref}
\usepackage[english]{babel}
\usepackage{placeins}
\usepackage{float}
\usepackage{caption}
\usepackage{ulem}
\usepackage[font=footnotesize,labelsep=quad,justification = justified]{caption}
\newtheorem{theorem}{Theorem}

\renewcommand{\qedsymbol}{$\blacksquare$}

\allowdisplaybreaks 

\begin{document}

\title{Chaotic renormalization group flow and entropy gradients over Haros graphs}\

\author{Jorge Calero$^{1,2}$, Bartolo Luque$^{1}$ and Lucas Lacasa$^{3}$ }
\email{lucas@ifisc.uib-csic.es}
\affiliation{$^1$Departamento de Matem\'atica Aplicada, ETSI Aeron\'autica y del Espacio, Universidad Polit\'ecnica de Madrid, Madrid, Spain;\\$^2$Signal and Communications Theory and Telematic Systems and Computing, Rey Juan Carlos University, Madrid, Spain \\$^3$Institute for Cross-Disciplinary Physics and Complex Systems IFISC (CSIC-UIB), Palma de Mallorca, Spain.}
\date{\today}

\date{\today}
\begin{abstract}

Haros graphs have been recently introduced as a set of graphs bijectively related to real numbers in the unit interval. Here we consider the iterated dynamics of a graph operator $\cal R$ over the set of Haros graphs. This operator was previously defined in the realm of graph-theoretical characterisation of low-dimensional nonlinear dynamics, and has a renormalization group (RG) structure. We find that the dynamics of $\cal R$ over Haros graphs is complex and includes unstable periodic orbits or arbitrary period and non-mixing aperiodic orbits, overall portraiting a chaotic RG flow. We identify a single RG stable fixed point whose basin of attraction is the set of rational numbers, associate periodic RG orbits with (pure) quadratic irrationals and aperiodic RG orbits with (non-mixing) families of non-quadratic algebraic irrationals and trascendental numbers. Finally, we show that the entropy gradients inside periodic RG orbits are constant. We discuss the possible physical interpretation of such chaotic RG flow and speculate on the entropy-constant periodic orbits as a possible confirmation of  a (quantum field-theoretic) $c$-theorem applied inside the invariant set of a RG flow.
\end{abstract}

 \maketitle


\newpage

\section{Introduction}

The importance of studying the phase diagram of the Renormalization Group (RG) operator of statistical mechanics and quantum field theories stems from the fact that the RG invariant set yields physical insight. Namely, attractive (stable) fixed points of the RG flow are connected with {\it observable} thermodynamical phases -- such is the case of the low temperature and high temperature fixed points in the RG analysis of the Ising model, which are linked to the existence of the ferromagnetic and paramagnetic phases respectively--. In the same vein, a repelling (unstable) fixed point of the RG flow yields a singularity in the free energy of the system and is therefore naturally linked with a (critical) phase transition -- such is the case of the Curie temperature in the Ising model, where the system is set at the phase transition between para and ferromagnetic phases. Moreover, the critical exponents that fully characterise the nature of the phase transition (e.g. the universality class) are connected by simple relations to the linearised RG transformation around the unstable fixed point.\\
Conventionally one expects that RG attractors consist of a finite number of isolated fixed points, and authors argue that monotonicity theorems such as the $c$-theorem in quantum field theories in principle precludes more complex RG flows. The quest for these exotic RG dynamics --and their physical interpretation-- was already proposed by Wilson \cite{Wilson, Wilson2}. In the realm of condensed matter physics, Derrida and co-authors \cite{Derrida} showed that the appearance of unstable periodic orbits in the RG flow yield the onset of an infinite number of singularities in the free energy, i.e. the onset of an infinite number of critical temperatures. However, the quest for physical systems whose RG flow brings about a complex phase diagram has remained largely elusive, despite a few notable examples originating in quantum field theory \cite{Wilson3, Morozov, bosschaert2022chaotic} and condensed matter \cite{Huse, Mckay} (see also \cite{jiang2021chaotic} and references therein for a recent example of a chaotic RG flow in a Potts model with long range interactions). From a physical viewpoint, observing chaotic RG flow suggests that the system exhibits qualitatively different statistical
behavior at different scales --a salient feature of e.g. spin glasses--.\\
In parallel, RG treatments have also been applied in the realm of nonlinear dynamics \cite{schuster2006deterministic}. When considering the transition to chaos in low-dimensional systems, the RG phase portrait allows to distinguish the different observable attractors and further exploration of the critical attractors provides a means to explore the universality of these transitions \cite{feigenbaum1978quantitative,feigenbaum1982quasiperiodicity,hu1982exact}. In such a context, the so-called visibility and horizontal visibility graphs \cite{LacasaVisibility, luque2009horizontal} have been introduced as graph-theoretical transformations of trajectories into graphs, opening the possibility of analysing different types of dynamics through the lens of graph theory and network science. In particular, a graph-theoretical renormalization operator $\cal R$ has been introduced, and its attractors obtained from the trajectories of low-dimensional chaotic systems have been used to provide a (graph-theoretical) RG description of the classical routes to chaos, i.e. the Feigenbaum, the quasiperiodic and the Pomeau-Manneville scenarios \cite{FeigenbaumGraphs, AnalyticalFeigenbaum, intermit, QuasiperiodicGraphs}.\\

Very recently, and inspired by the application of horizontal visibility graphs (HVGs) to the quasiperiodic route, an independent graph-theoretic description of real numbers --rationals and irrationals-- has been proposed, and the objects thereby associated to each number have been coined as Haros graphs \cite{HarosPaper}.  Each real number is, accordingly, bijectively related to a different Haros graph, thus building a bridge between number theory and graph theory. Since Haros graphs are a subset of HVGs, the same RG operator $\cal R$ previously defined on HVGs could be applied to Haros graphs, even if the latter are not in principle naturally related to the trajectories of any underlying dynamical system. In this paper we thereby define and explore the iteration dynamics of such RG operator over the set of Haros graphs. Interestingly, we find that the iterated dynamics $\cal R$ is extremely complex, showing distinctive properties of a chaotic RG flow. Furthermore, we show that different classes of irrational numbers naturally cluster together under the action of the RG operator on Haros graphs, defining classes of periodic orbits for different families of quadratic irrationals, whereas other algebraic and trascendental irrationals cluster together in families of aperiodic RG orbits. To provide analytical insight into the dynamics of the RG operator, we rigorously prove that the graph $\cal R$ is topologically conjugate to the interval modified Farey map \cite{Adamczewski, AMM}, a chaotic map itself topologically conjugate to the shift map via Minkowski's question mark function.
We end up exploring entropic flows along the RG action, following classical works both in QFT \cite{zamolodchikov1986irreversibility, cardy1988there} and nonlinear dynamics \cite{Robledo} that study the irreversible nature of RG flow and its quantification.\\

The rest of the paper is as follows: in section \ref{sec:preliminaries} we provide the necessary preliminaries, introduce Haros graphs and their main properties. In section \ref{sec:R} we introduce the graph operator $\cal R$, show that it performs a topological thinning of the system, and prove that when applied to an arbitrary Haros graph, the resulting graph is again a Haros graph, thereby establishing that its iteration over the set of Haros graphs is a well-defined map and that Haros graphs are renormalizable. In section \ref{sec:dyn} we prove that $\cal R$ is topologically conjugate to an interval map, explicitly finding the (real-valued) coupling constant space where the RG flow takes place and an analytical expression for the RG flow, in terms of a one-dimensional interval map. We leverage on such topological conjugacy to analytically explore the invariant set of $\cal R$, finding a hierarchy of fixed points, periodic orbits (which we associate to quadratic irrationals) and families of non-mixing aperiodic orbits, including algebraic aperiodic orbits (AAO) and trascendental aperiodic orbits (TAO). In section \ref{sec:ent} we explore entropic gradients along the RG flow, both towards the stable fixed point and inside the RG periodic orbits, finding, notably, an entropic function which remains constant throughout each periodic orbit. In section \ref{sec:con} we conclude.


\begin{figure*}[h!]
\includegraphics[width=1\textwidth]{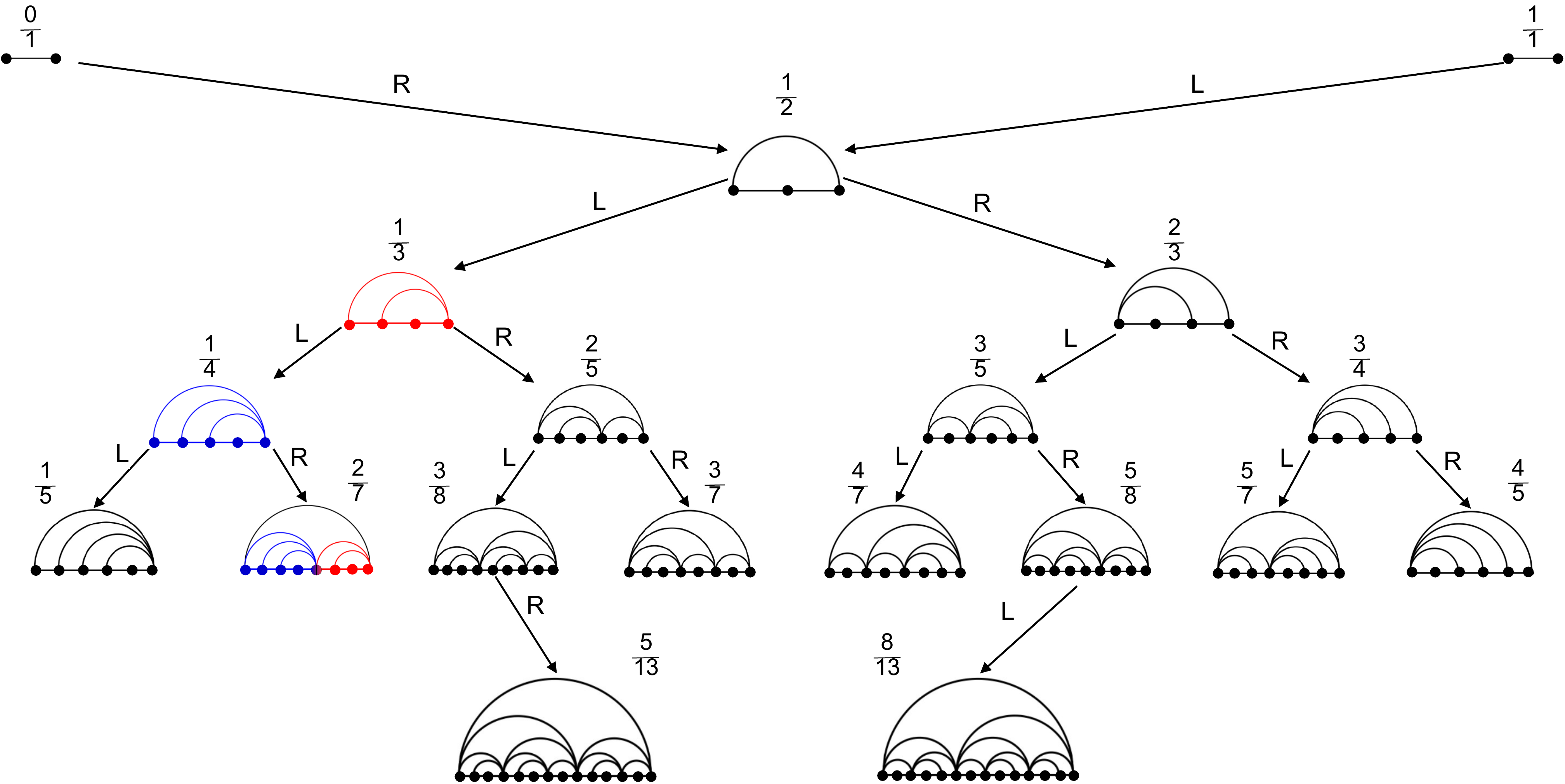}
\caption{Six levels of the Haros graph binary tree with the Haros graphs $G_{p/q}$ associated with the corresponding rational fractions $p/q$ (for space reasons only two of these are shown at the sixth level). On the left, Haros graph $G_{2/7}$ building as concatenation of $G_{1/4}$ (blue color) and $G_{1/3}$ (red color). Therefore $G_{2/7} = G_{1/4} \oplus G_{1/3}$, where the concatenation operator merges the right-extreme node of $G_{1/4}$ with the left-extreme node of $G_{1/3}$, and a link between the new extreme nodes are created. Moreover, any Haros graph can be expressed as a concatenation of two adjacents Haros graphs $G = G_L \oplus G_R$.}
\label{Fig_1}
\end{figure*}

\section{Preliminaries}
\label{sec:preliminaries}
In a recent work \cite{HarosPaper}, a new graph-theoretical representation of the real numbers was introduced. Here we recall the basic facts, but we urge the reader to dive deep in \cite{HarosPaper} for more details.\\
A canonical representation of the unit interval $[0,1]$ was provided through the so-called Farey sequences
$$\mathcal{F}_{n} = \left\lbrace\frac{p}{q} \in [0,1]:\ 0\leq p \leq q \leq n, \ (p,q) = 1 \right\rbrace,$$
where $p$ and $q$ are natural numbers.\\
\noindent The Farey sequences $\mathcal{F}_{n}$ can be generated iteratively via the mediant sum $\frac{p}{q}\oplus\frac{p'}{q'}:=\frac{p+p'}{q+q'}$ of each consecutive par of elements in  $\mathcal{F}_{n-1}$. Thus, a classical way of hierarchically displaying the elements of ${\cal F}_n$ is using the so-called {\it Farey binary tree} (Fig. \ref{Fig_1} illustrates the  first six levels of Farey binary tree).\\
The abovementioned graph-theoretical representation is based on the definition of so-called Haros graphs set $\mathcal{G}$ \cite{HarosPaper}, obtained recursively via an initial graph $G_0=K_2$, and a concatenation  graph-operation $\oplus$ (Fig. \ref{fig:Fig2_FirstFarey} sketches the performance of the graph operator). The construction of Haros graphs replicates the inflationary process of Farey fractions giving a one-to-one correspondence between Farey sequences $\mathcal{F}_n$ and the Haros graph set (see Fig.\ref{Fig_1} for an illustration).  The bijection can be extended to the real interval as  $\lim_{n\to \infty} \mathcal{F}_n = [0,1]$, i.e., $\tau: [0,1] \to \mathcal{G}$ is a bijection defined as $\tau(x) = G_x, \ \forall x\in[0,1]$. By construction, rational numbers $x=p/q$ are associated to finite Haros graphs, whereas irrational numbers map to infinite Haros graphs (see section III in \cite{HarosPaper} for more details). As the set $\mathcal{G}$ can be displayed as a (Haros graph) binary tree, any arbitrary Haros graph can be shown to be reached via a symbolic binary path. Such symbolic paths in the Haros graph binary tree --illustrated in Fig. \ref{Fig_1}--  are themselves intimarely related with the continued fraction expression of the associated number $x$ of the Haros graph $G_x$.

\begin{figure}[h!]
\begin{center}
\includegraphics[width=0.8\textwidth , trim=0cm 17cm 0cm 5cm,clip=true]{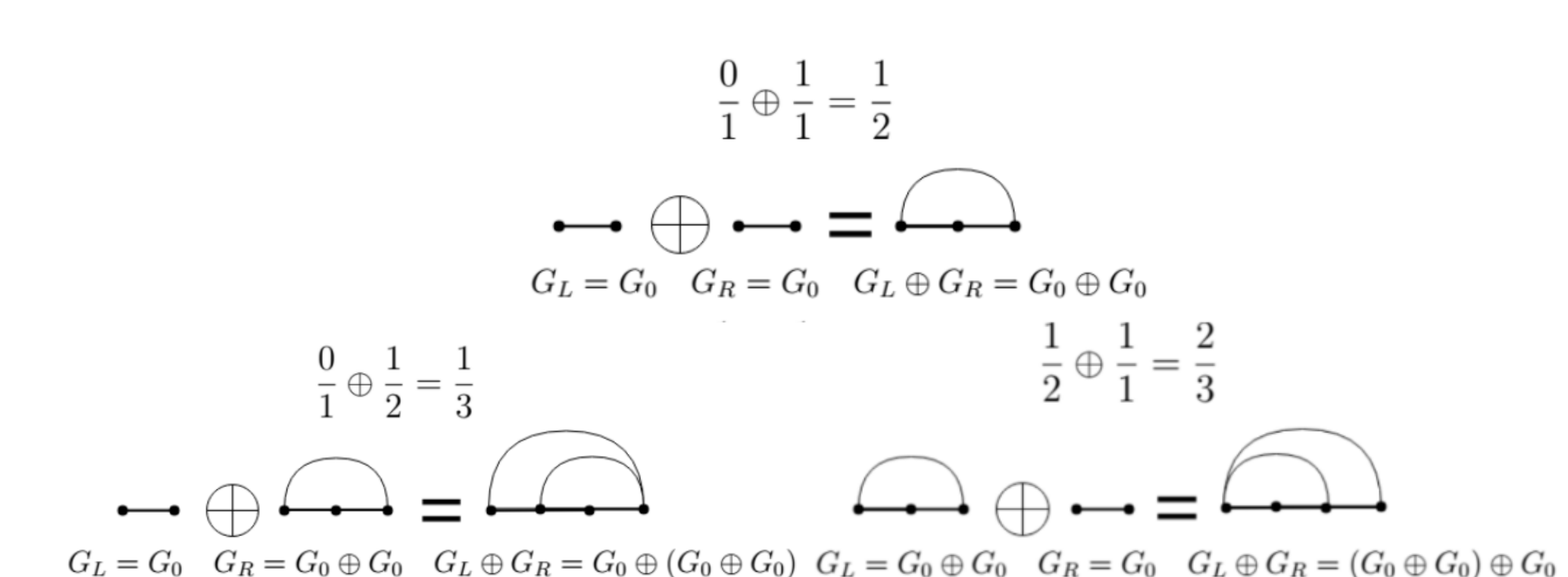}
\caption{ \textbf{First line:} Construction of $G_1 = G_0 \oplus G_0$. The graph concatenation operator merges the extreme nodes: labelling the nodes of $G_0$ as $[1,2]$, the extreme nodes are $2$ for $G_L$ and $1$ for $G_R$. The graph obtained by the concatenation has three nodes and the external nodes ($1$ and $3$ in $G_L \oplus G_R$) are linked by a new vertex.  \  \textbf{Second line:}  Construction of the two possible concatenations between $G_0$ and $G_1$, i.e., $G_0 \oplus G_1$ and $G_1 \oplus G_0$.}
\label{fig:Fig2_FirstFarey}
\end{center}
\end{figure}

\noindent Among others, a key topological property of $G_x$ that was considered in \cite{HarosPaper} to analyse the resulting representation of the reals was the degree distribution $P(k,x)$, interpreted as the probability (normalized frequency) that a node chosen randomly in $G_x$ has degree $k$.
As a matter of fact, the first values of the degree distribution verifies:
\begin{equation*}
P(k,x < 1/2)=\left\{
\begin{array}{ll}
x,  & k=2 \\
1 - 2x, & k=3 \\
0, & k=4
\end{array}%
\right.
; P(k,x > 1/2)=\left\{
\begin{array}{ll}
1 - x,  & k=2 \\
2x - 1, & k=3 \\
0, & k=4.
\end{array}%
\right.
\end{equation*} 
In particular, in \cite{HarosPaper} it was shown that the graph-entropy function $S(x) := - \sum_{k\geq 2} P(k;x)\cdot \log P(k;x)$ to be a fractal, continuous everywhere but differentiable nowhere function. Its global maximum was reached at the Haros graph associated to reciprocal of the Golden number $x=\phi^{-1}$ (and its symmetric $x=1 - \phi^{-1}$), whereas local maxima in each interval $[0,1/n]$ are reached by the Haros graphs associated to an affine transformation of Golden number known as noble numbers: quadratic irrationals with a continued fraction that becomes an infinite sequence of $1$s \cite{Hardy, Schroeder}. This result can be extended to every quadratic irrational, which Haros graphs present a degree distribution $P(k,x)$ with  an exponential tail. Moreover, the local minimums of $S(x)$ are associated to rational numbers.

\section{An RG operator \texorpdfstring{$\mathcal{R}$}{R1} for Haros graphs}
\label{sec:R}


The graph operator $\cal R$ was originally defined in the context of horizontal visibility graphs (HVGs) \cite{FeigenbaumGraphs, AnalyticalFeigenbaum}, as a coarse-graining in graph space. Given a HVG $G$, ${\cal R}(G)$ is a graph where all the inner nodes with degree $k=2$ are decimated, or, equivalently,  where every two nodes that share an adjacent node with degree $k=2$ are coarse-grained into a block node. By applying this topological thinning iteratively to HVGs extracted from trajectories in the Feigenbaum scenario --canonically exemplified by the logistic map $x_{t+1}=\mu x_t(1-x_t)$--, it was shown \cite{FeigenbaumGraphs, AnalyticalFeigenbaum} that $\cal R$ was akin to a Kadanoff-like decimation of the configuration of a spin system. 
Indeed, iteratively applying $\cal R$ in the period-doubling bifurcation cascade leads to a flow in the space of map parameter $\mu$, the only relevant variable. Interestingly, such graph-theoretical RG technique was successfully used to characterise not only Feigenbaum scenario, but other canonical low-dimensional routes to chaos, including the quasiperiodic route \cite{QuasiperiodicGraphs} or Pomeau-Manneville (intermittency) routes \cite{intermit}, where the classic RG flow diagram (a non-trivial, unstable RG fixed point denoting the edge of chaos, and two trivial, stable RG fixed points denoting the regular and chaotic regions) was systematically found for $\cal R$. 
Furthermore, e.g. in the Feigenbaum scenario the macroscopic properties of graphs ensembles under iteration of $\cal R$ reach asymptotically a stationary solution above and below the edge of chaos (i.e. the flow reaches a stable fixed point), whereas at the edge of chaos (e.g. the accumulation point in the Feigenbaum scenario) the graph is invariant under $\cal R$. Any small deviation from the edge of chaos yields a graph whose RG trajectory flows away from it, hence the fixed point is unstable.\\
\begin{figure}[htb]
\includegraphics[width=0.7\textwidth, trim=0cm 2cm 0cm 0cm,clip=true]{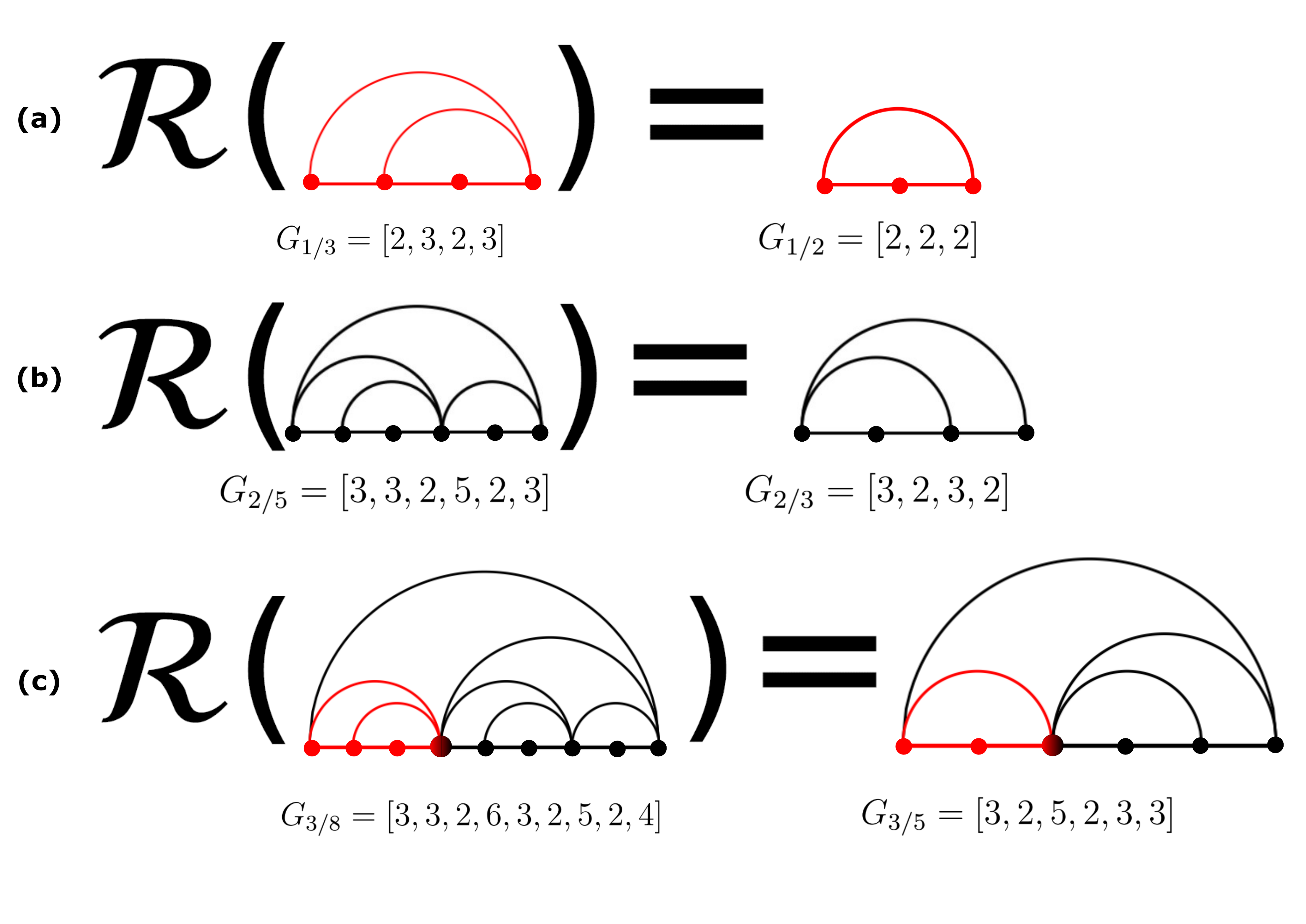}
\caption{Three examples of the renormalization operator $\mathcal{R}$. Given a graph $G$, ${\cal R}(G)$ topologically decimates $G$ by removing all $k=2$ nodes and gluing together the adjacent nodes. (Panel a) Application of $\cal R$ to the  Haros graphs $G_{1/3}$, whose degree sequence is $[2,3,2,3]$. Applying $\mathcal{R}$, the inner node with degree $k=2$ and its links are removed. The resulting graph is another Haros $G_{1/2} = [2,2,2]$.  (Panel b)  The degree sequence of the Haros graph $G_{2/5}$ in this example is ${\bf k}=[3, 3, 2, 5, 2, 3]$, while the one after application of $\mathcal{R}$ is ${\bf k}=[3, 2, 3, 2]$. (Panel c)  Applying the renormalization to $G_{3/8} = G_{1/3} \bigoplus G_{2/5}$, the Haros graph obtained is the concatenation of  $\mathcal{R}(G_{1/3})$ and $\mathcal{R}(G_{2/5})$, i.e., $\mathcal{R}(G_{3/8}) = \mathcal{R}(G_{1/3}) \bigoplus \mathcal{R}(G_{2/5}) =  G_{1/2} \bigoplus G_{2/3} = G_{3/5}$.}
\label{Fig_2}
\end{figure}
\noindent It is important to highlight that $\cal R$ is more similar to a real-space RG transformation (actually, a decimation) on the graph \cite{Newman}, than to other graph renormalization methods such as box-covering renormalization \cite{Song2005, Song2006, Radicchi}. Since Haros graphs are a subset of the set of HVGs \cite{HarosPaper}, it is legitimate to apply the $\cal R$ operator to Haros graphs, even if Haros graphs are not in principle the HVGs of a pre-specified trajectory extracted from some map. For illustration, in Figure \ref{Fig_2} we provide examples of the action of $\cal R$ for different Haros graphs.

\subsection{\texorpdfstring{$\mathcal{R}$}{R1} is a map of the Haros graph set onto itself, i.e., $\cal G$ is renormalizable}
\noindent The first important question is whether the set of Haros graphs is closed under the action of $\cal R$, i.e. if $G$ is an arbitrary Haros graph, whether ${\cal R}(G)$ is also Haros. From a physical jargon point of view, this would imply that the theory is renormalizable, as the system $G\in {\cal G}$ can be expressed in a similar fashion (i.e. ${\cal R}(G)$ is also Haros) but with a different `coupling constant'. Below we prove that such is indeed the case.


\begin{theorem}
\label{Tma1}
Let $\cal G$ be the set of Haros graphs.  Then $\forall G\in {\cal G}$, ${\cal R}(G) \in \mathcal{G}$. Moreover, if $G$ and $G'$ are adjacent in the Haros graph tree, then ${\cal R}(G)$ and ${\cal R}(G')$ are adjacent as well.
\end{theorem}

\begin{proof}
 The key of the proof resides on the fact that the concatenation operator $\oplus$ and the graph operator $\mathcal{R}$ are associative. It has to be noted that if $G$ is a Haros graph, by construction, then $G = G_L \bigoplus G_R$, where $G_L, G_R$ are also adjacent Haros graphs. Hence, we will have $\mathcal{R}(G) = \mathcal{R}( G_L \bigoplus G_R ) = \mathcal{R}( G_L ) \bigoplus \mathcal{R}( G_R )$.\\ 
 Let be  $\ell_i$ the set of new irreducible fractions that are created in the Farey tree in the level $i$ and $\mathcal{L}_n=\cup_{i \in n} \ell_{i}$ the set of all irreducible fractions that are created in the Farey tree up to the level  $n$.  We will prove the theorem by induction over the levels of the Farey binary tree: \\
\noindent {\it Case} $n = 1$: 
If $G_0 = G_{0/1}$ or $G_0 = G_{1/1}$, then $\mathcal{R}(G_0) = G_0 \in \mathcal{G}$.\\
\noindent {\it Case} $n = 2$: 
$\mathcal{R}(G_{1/2} = [2,2,2]) = [1,1] =  G_0 \in \mathcal{G}$.\\
\noindent {\it Case} $n = 3$: 
$\mathcal{R}(G_{1/3} = [2,3,2,3]) = \mathcal{R}(G_0 \bigoplus (G_0 \bigoplus G_0)) =  [2,2,2] = G_{1/2}\in \mathcal{G}$ 
and $\mathcal{R}(G_{2/3} = [3,2,3,2]) = \mathcal{R}((G_0 \bigoplus G_0) \bigoplus G_0) =  [2,2,2] = G_{1/2}\in \mathcal{G}$.\\
\noindent {\it Inductive hypothesis:} Suppose that for all Haros graph $G = G_{p/q}$, where  $p/q \in \mathcal{L}_n$, then  $\mathcal{R}(G) \in \mathcal{G}$,  and that taking two adjacent Haros graphs $G, G'$, the images from the renormalization are adjacent Haros graphs as well. 

\noindent Let be a Haros graph $G = G_{p/q}$ where $p/q \in \ell_{n+1}$. By definition of Haros graphs, $G = G_L \bigoplus G_R$, where $G_L, G_R \in \mathcal{L}_n$. The merging node  what concatenate $G_L$ and $G_R$ has degree $k \geq 3$ by construction. Hence, the renormalization operator does not remove this node and clearly we have $\mathcal{R}(G) = \mathcal{R}( G_L \bigoplus G_R ) = \mathcal{R}( G_L ) \bigoplus \mathcal{R}( G_R )$. As $G_L$ and $G_R$ can be concatenated to obtain $G$, then $G_L$ and $G_R$ are adjacent Haros graph. By the inductive hypothesis, $\mathcal{R}( G_R )$ and $\mathcal{R}(G_L) $ are also adjacent Haros graph. Thus implies that the concatenation $\mathcal{R}( G_L ) \bigoplus \mathcal{R}( G_R )$ produces an Haros graph, i.e.,   $\mathcal{R}(G)\in \mathcal{G}$. This concludes the proof.
\end{proof}

According to the previous theorem, $\cal R$ is a map of the set of Haros graphs into itself and therefore  we say that $\cal G$ is renormalizable:  $\cal R$ can be applied iteratively to cover regions of the set of Haros graphs. Accordingly,  the iterations of $\cal R$ over $\cal G$ induces an RG flow. It remains to be understood, however, what is the specific analog of the `coupling constant space' in which such flow effectively develops, so as to explore the dynamics of $\cal R$ and its long-term dynamical behavior.  This is achieved in the next section.

\section{The Dynamics of \texorpdfstring{$\mathcal{R}$}{R1} induce a classification of real numbers}
\label{sec:dyn}

\subsection{Renormalization of Finite Haros}
We start this section by shedding some light into how $\cal R$ transforms finite Haros graphs.\\
As we have mentioned in section II, rational Haros graphs $G_{p/q}$ are finite graphs and, with $p/q \leq 1/2$, we know that $P(k = 2,p/q) = p/q, P(k = 3,p/q) = 1 - 2\cdot p/q$, and $P(k = 4,p/q) = 0$ \cite{HarosPaper}. In other words, the total number of nodes of $G_{p/q}$ is $q$ and the number of nodes with $k=2$ is just $p$. Hence, ${\cal R}(G_{p/q})$ is necessarily a Haros graph with $q-p$ nodes. Since indeed ${\cal R}(G_{p/q})$ is also Haros and is finite, then there exists $p'$ such that 
\begin{equation}
    {\cal R}(G_{p/q})=G_{p'/(q-p)}.
    \label{eq:scaling1}
\end{equation}
The same reasoning can be applied if $p/q \geq 1/2$, where from $P(k = 2,p/q) = (q-p)/q$, we get ${\cal R}(G_{p/q})=G_{p'/p}$ with $p'<p<q$. Thus, repeated iterations of $\cal R$ will smoothly decimate all degree-2 nodes, i.e. after a finite number $m$ of iterations, we necessarily reach $\mathcal{R}^{(m)}(G_{p/q}) = G_{1/1}={\cal R}(G_{1/1})$.  An illustration for such flow is presented in Fig.\ref{Fig_3}.\\
In other words, $G_{1/1}$ appears as a stable fixed point of the operator $\mathcal{R}$, whose basin of attraction consists of the set of all rational (i.e., finite) Haros graphs. Since a graph $G_{p/q}$ is bijectively related to the rational $p/q$, the iteration of $\cal R$ is also inducing an associated sequence of rational numbers.\\
First, observe that Eq.\ref{eq:scaling1} is suggesting a scaling relationship where the action of $\cal R$ is renormalizing the parameter $p/q\to p'/(q-p)$. If $p'$ was precisely controlled, then such eq. would be giving us the coupling constant flow of the RG action. Second, Eq. \ref{eq:scaling1} is an observation for rational `coupling constants': it is still unclear if such equation also holds for arbitrary irrational numbers. In the next subsection we address both issues.

\begin{figure*}[htbp!]
\includegraphics[width=1\textwidth]{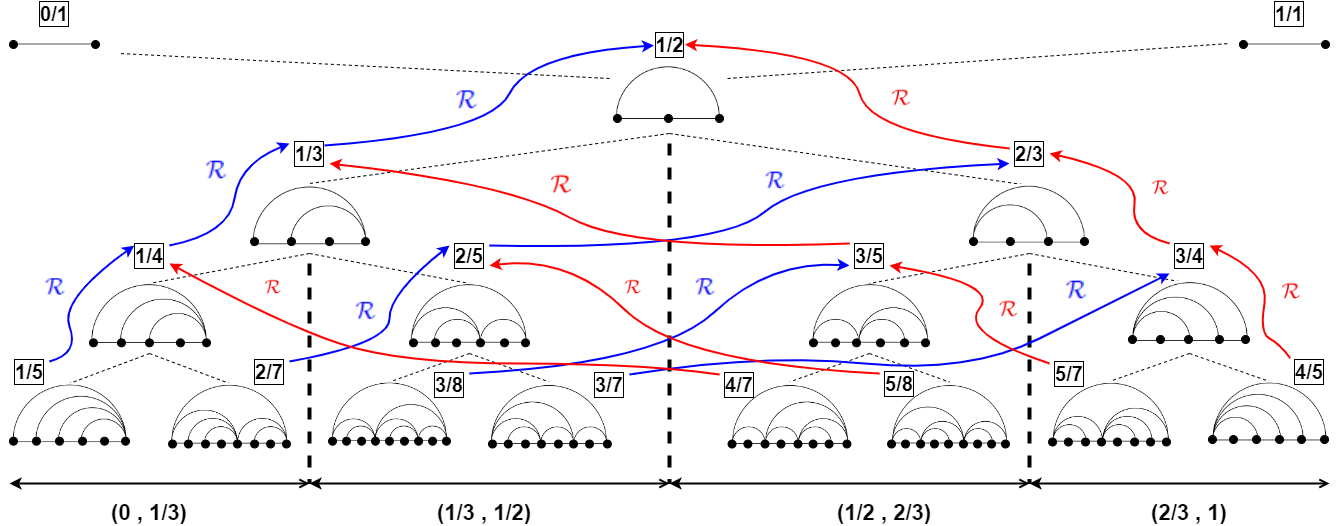}
\caption{Haros graph Tree with renormalization dynamics. The flow depicted in blue represent the left branch $T_1(x) = x/(1-x)$ displacing the intervals $[1/(n+1), 1/n]$ to $[1/n, 1/(n-1)]$ for $n \geq 2$. In the same manner, the right branch --depicted in red-- $T_2(x) = (2x-1)/x$ moving the intervals $[(n-2)/(n-1) ,(n-1)/n]$ for $n \geq 3$. In both cases, if $G$ is located at level $\kappa \geq 1$, then $\mathcal{R}(G)$ is located at level $\kappa - 1$.}
\label{Fig_3}
\end{figure*}

\subsection{A topological conjugacy for \texorpdfstring{$\mathcal{R}$}{R1}}
What kind of trajectory is generated under the iteration of $\cal R$ on $G_x$, when $x$ is irrational and our graph is infinite and thus seemingly unaccessible to numerical exploration?   A strategy is to find an operator $T(x)$ that acts on $[0,1]$, in such a way that the results of this operator acting on $x \in [0,1]$ can be related to ${\cal R}(G_x)$, thereby extending Eq.\ref{eq:scaling1} to arbitrary values $x$. We postulate that such operator is the following interval map $T$:
\begin{equation}
T(x)=\left\{
\begin{array}{ll}
T_{1}(x) = \frac{x}{1-x}  &  \textrm{if} \ x \leq 1/2 \\
T_{2}(x) = \frac{2x-1}{x} & \textrm{if} \ x > 1/2 \\
\end{array}%
\right.
\label{Eq_T}
\end{equation}
\noindent Formally, we are claiming that the maps $T$ and $\mathcal{R}$ are topologically conjugate. This means that there exists a homeomorphism $f$ which verifies $f \circ T = \mathcal{R} \circ f$. The homeomorphism that gives us this assertion is the application $\tau$, which identify number $p/q$ with the associated Haros graph $G_{p/q}$ i.e., is defined as $\tau(p/q) = G_{p/q}$. If such structure holds, we'd obtain the following commutative diagram:

\begin{center}
\begin{tikzcd}
\mathcal{G} \arrow{r}{\mathcal{R}}  & \mathcal{G}  \\
     \left[0,1\right] \arrow{r}{T} \arrow[swap]{u}{\tau} & \left[0,1\right]  \arrow{u}{\tau}
\end{tikzcd}
\end{center}
Below we state and prove that such is indeed the case.

\begin{theorem}
\label{Tma2}
$\forall x \in [0,1], \mathcal{R}(G_{x}) = G_{T(x)}$. In other words,  $T$ and $\mathcal{R}$ are maps topologically conjugate. 
\end{theorem}

\begin{proof}
Applying induction over the levels of the Farey binary tree: \\
{\it Level} $n = 1$: 
$\mathcal{R}(G_{0/1}) = G_{0/1}$, and $T(\frac{0}{1}) = \frac{0}{1 - 0} = \frac{0}{1}$. The other special Haros graph verifies that
$\mathcal{R}(G_{1/1}) = G_{1/1}$ and $T(\frac{1}{1}) = 1 - \frac{1 - 1}{1} = \frac{1}{1}$. \\
{\it Level} $n = 2$: 
$\mathcal{R}(G_{1/2}) = G_{1/1}$, and $T(\frac{1}{2}) =\frac{1/2}{1 - 1/2} = \frac{1}{1}$. \\
{\it Level} $n = 3$: 
$\mathcal{R}(G_{1/3}) = G_{1/2}$ and $T(\frac{1}{3}) = \frac{1/3}{1 - 1/3} = \frac{1}{2}$. $\mathcal{R}(G_{2/3}) = G_{1/2}$ and $T(\frac{2}{3}) = 1 - \frac{1 - 2/3}{2/3} = \frac{1}{2}$. \\
{\it Inductive hypothesis:} The result is verified for all Haros graph $G = G_{p/q}$, where  $p/q \in \mathcal{L}_n$. 

\noindent Let be a Haros graph $G_{p/q}$ with $p/q \in \ell_{n+1}$. It is clear that $G =G_{p_1/q_1} \bigoplus G_{p_2/q_2}$, where $p_1/q_1, p_2/q_2 \in  \mathcal{L}_n$, and $\frac{p}{q} = \frac{p_1}{q_1} \oplus \frac{p_2}{q_2} = \frac{p_1 + p_2}{q_1 + q_2}$ . In virtue of Theorem \ref{Tma1}, we know that $\mathcal{R}(G_{p/q}) = \mathcal{R}( G_{p_1/q_1} ) \bigoplus \mathcal{R}( G_{p_2/q_2} )$. We distinguish two possibilities:
\begin{enumerate}
    \item If $G_{p/q}$ with $p/q \leq 1/2$, then  $p_1/q_1, p_2 /q_2 \leq 1/2$. Using the induction hypothesis: $$\mathcal{R}(G_{p_1/q_1}) = G_{T_1(p_1/q_1)},$$  $$\mathcal{R}(G_{p_2/q_2}) = G_{T_1(p_2/q_2)}.$$ The concatenation give us an Haros graph $\mathcal{R}(G_{p/q})$ associated to the number:
    \begin{eqnarray*}
    T_1\left(\frac{p_1}{q_1}\right) \oplus T_1\left(\frac{p_2}{q_2}\right) & = & \frac{p_1/q_1}{1 - p_1/q_1} \oplus \frac{p_2/q_2}{1 - p_2/q_2} = \frac{p_1}{q_1 - p_1} \oplus \frac{p_2}{q_2 - p_2} = \\
   & = & \frac{p_1 + p_2}{q_1 - p_1 + q_2 - p_2} = \frac{p}{q-p}=  T_1\left(\frac{p}{q}\right)
    \end{eqnarray*}

    \item If $G_{p/q}$ with $p/q \geq 1/2$, then  $p_1/q_1, p_2 /q_2 \geq 1/2$. Again, by the induction hypothesis: $$\mathcal{R}(G_{p_1/q_1}) = G_{T_2(p_1/q_1)},$$  $$\mathcal{R}(G_{p_2/q_2}) = G_{T_2(p_2/q_2)}.$$ In this case, the resulting Haros graph $\mathcal{R}(G_{p/q})$ is associated to the number:
    \begin{eqnarray*}
   T_2\left(\frac{p_1}{q_1}\right) \oplus T_2\left(\frac{p_2}{q_2}\right)& = & \frac{2\cdot p_1/q_1 -1 }{p_1/q_1} \oplus \frac{2\cdot p_2/q_2 - 1}{p_2/q_2} = \frac{2p_1 - q_1}{p_1} \oplus \frac{2p_2 - q_2}{p_2} = \\
   & = & \frac{2(p_1 + p_2) - (q_1 + q_2)}{p_1 + p_2} = \frac{2p - q}{p} =T_2\left(\frac{p}{q}\right).
    \end{eqnarray*}
\end{enumerate}
This concludes the proof.
\end{proof}
\noindent Two comments are in order:\\
First, this topological conjugacy is revealing the nature of the analog to coupling constants flow in a standard RG treatment. The iteration of $\cal R$ in Haros graph space is inducing a flow in $x\in[0,1]$, the coupling constant space. Such flow is, therefore, the RG flow.\\ 
Second, the explicit topological conjugacy that we found guarantees that the dynamical properties of $\cal R$ are induced by those of the map $T$, and therefore in order to investigate the dynamics of $\cal R$ over $\cal G$ one simply needs to consider the dynamical properties of $T$ over $[0,1]$, as topological conjugacy preserve the structure of the invariant set (i.e. fixed points map into fixed points, periodic orbits map into periodic orbits, etc): 
if for instance an element $x \in [0,1]$ is a fixed point of $T$, then $G_{x}$ is a fixed point of $\cal R$, and similarly an orbit of $x$ defined as $\left\lbrace T^n(x)\right\rbrace_{n = 0}^{\infty}$ corresponds with an orbit $\left\lbrace \mathcal{R}^n(G_x)\right\rbrace_{n = 0}^{\infty}$ in the Haros graph space, and so on.

\subsection{RG invariant set: A dynamical classification of the reals}

\noindent Let us start the exploration of $\cal R$ by giving some general comments about $T$. First, the algebraic operator $T$ has been defined as \textit{Farey modified map} \cite{IsolaDynamic, Saito, Bonnano_Isola}. It can be shown that $T$ is topologically conjugate to the well-known shift map $D(x)= 2x \ (\textrm{mod} 1)$ via the so-called Minkowski's question mark function $?(x)$, i.e. $(? \circ T)(x) = (D \circ ?)(x)$, (see Fig. \ref{Fig_4}) \cite{Bonnano_Isola}. Hence, $T$ is actually a expansive, orientation preserving, piece-wise analytic endomorphism with positive Lyapunov exponent (as the Lyapunov of $D(x)$ is $\log(2)$ \cite{Beirami}). As the shift map, $T$ is not invertible: for any $x$, $T(x)$ has two possible pre-images: 
$$\left\lbrace \frac{x}{1+x}, \frac{1}{2-x} \right\rbrace.$$
The pre-image $\frac{x}{1+x}$ coincides with the operator $F(x)$ defined in \cite{HarosPaper}, which acts in continued fraction as $x=[a_1, a_2, ...] \rightarrow F(x)=x/(1+x)=[a_1 + 1, a_2, ...]$. The continued fraction expression $x = [a_1, a_2, ...]$ codifies the symbolic path followed by the Haros graph, which is $L^{a_1} R^{a_2} L^{a_3} ... $ where $L^a$ means a sequence of $a$ Ls (in the rational case, the last symbol has an index $a_n -1$). Indeed, the action of the map $T$ over continued fraction has been presented in \cite{Saito}, where the author exposed that if $x = [a_1,a_2 ...]$: 

\begin{equation}
T([a_1, a_2, a_3, ...])=\left\{
\begin{array}{ll}
[a_1 - 1, a_2, ...] & , a_1 \geq 2  \\

[1, a_2 - 1, a_3, ...] & ,a_1 =1 \; \textit{ and } \; a_2 \geq 2 \\

[a_3 + 1, a_4, ...] & , a_1 = a_2 =1.
\end{array}%
\right.
\label{OpT_contfrac}
\end{equation} 

\begin{figure}[htbp]
\includegraphics[width=0.4\columnwidth]{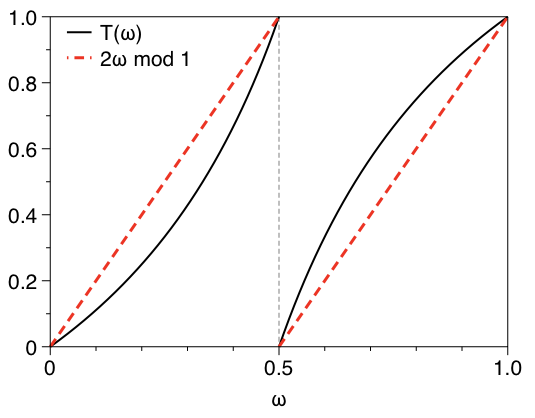}
\caption{Illustration of the map $T(w)$ and the Bernoulli shift $2w \mod 1$. These maps are topologically conjugate.}
\label{Fig_4}
\end{figure}

We now explore in detail the invariant set of $T$, and show how it induces a dynamical classification of real numbers in $[0,1]$ into families (and thus, an equivalent classification of the Haros graphs holds):
\begin{enumerate}
\item {\bf Fixed Points:} The study of fixed points --solving $T(x^*)=x^*$-- unveils that the only fixed points are $x^*=0,1$, related to $G_{0/1}$ and $G_{1/1}$ respectively. These are indifferent fixed points as $T'(x^*)=1$ ($x^*=0$ is a weakly repellent fixed point whereas $x^*=1$ is weakly attractive). Ergodic and spectral properties of $T$ have been studied \cite{Bonnano_Isola}.\\
Similarly to the shift map, $x=1$ is an attractor for all rational initial conditions (fixed point), i.e. $\forall x=p/q, \ \exists m \geq0: \ T^{(m)}(x)=1$, and equivalently, $ \mathcal{R}^{(m)}(G_{p/q})=G_{1/1}$. This is in good agreement to the analysis performed in Section IVA. Note that one can easily give a pictorial interpretation of this process in graph space (see Fig. \ref{Fig_3}): for any rational value, its associated Haros graph $G_{p/q}$ flows, under action of $\cal R$, upstream one level in the Haros graph Tree, eventually reaching the top of the Haros graph Tree $G_{1/1}$ which is a (stable) fixed point.

\item {\bf Periodic Orbits:} The so-called quadratic irrational numbers are those irrational values $x$ which are solutions of a quadratic equation with integer coefficients $a,b,c$:
$$ax^2 + bx +c =0.$$
Quadratic irrational numbers have continued fraction expansions that {\it eventually} repeat \cite{Khinchin}, of the form $[c_1, ... , c_r, \overline{a_1, ... , a_n}]$, where the terms $c_i$ are a finite transient and the terms $a_j$ the periodic part. That is to say, we can first distinguish the {\it pure} quadratic irrationals, of the form $x = [\overline{a_1, ... , a_r}]$. { According to Eq. \ref{OpT_contfrac}, as $T(x)$ acts as a shift operator over the continued fraction, these initial conditions generate periodic orbits under the action of $T(x)$ of period $m\geq 2$, i.e. $T^{(m)}(x) = x$}. By virtue of theorem 2, this invariant set of periodic orbits is the first exotic property of the RG flow, i.e. ${\cal R}^{(m)}(G_x) = G_x$.

\noindent Furthermore, it is easy to see that all these (unstable) periodic orbits of period $m>1$ are indeed formed {\it only} by pure quadratic irrationals. To justify this, one can simply decompose $T$ (see Eq. \ref{Eq_T}) into its two branches, $T_1(x), T_2(x)$, and observe that both branches are, actually, particular cases of a M\"obius transformation of the type
\begin{equation}
    M(x)=\frac{A+Bx}{C+Dx},
    \label{eq:Mobius}
\end{equation}
where $A\cdot D - B\cdot C \neq 0$. Now, the condition for $x$ to be a periodic point belonging to a cycle of length $m$ is $T^{(m)}(x)=x$. This equation admits at most $2^m$ different compositions of $T_1$ and $T_2$ of the type $T_i\circ T_i \circ \dots T_i$ where $T_i$ is either $T_1$ or $T_2$ (for a specific initial condition $x$, there is actually only one assignment). However the composition of two M\"obius transformations is again a M\"obius transformations, meaning that for an arbitrary $m$, solving the fixed point equation  $T^{(m)}(x)=x$ will always reduce to find the roots of a family of quadratic polynomials, whose solutions are by definition quadratic algebraic numbers, i.e. either rationals or quadratic irrationals. Since we know that all rationals converge under iteration to the attractive fixed point, then the only possible solutions of  $T^{(m)}(x)=x$ are quadratic irrationals. Additionally, changing $m$ essentially changes the coefficients of the M\"obius transformation $T^{(m)}(x)$, in such a way that this equation admits solutions $\forall m$: hence there exists periodic orbits of {\it arbitrarily} large period, and these consist of (and only of) pure quadratic irrationals.

In contrast, under the action of $T$, the {\it non pure} quadratic irrationals (those that have a transient in its continued fraction expansion) form trajectories which, after $\alpha = c_1 + ... + c_r$ applications of $T$, converge into the orbit generated by $[\overline{a_1, ... , a_r}]$ . In other words, each particular pure quadratic irrational is a stable attractor of the subset of non-pure quadratic irrationals with the same periodic part, i.e. the periodic orbit associated to a given set of pure quadratic irrationals is indeed a limit cycle.  
\\ Additionally, observe that the set of rational numbers is dense in $[0,1]$. This necessarily means that these periodic orbits are unstable, as any arbitrary small deviation from the (pure quadratic irrational) initial condition will collide with a rational number which, under the action of $T$ will converge to the fixed point, moving away from the neighborhood of the periodic orbit. Moving back to $\cal R$, we conclude that the RG flow in Haros graph space has an infinite number of periodic orbits which are at the same time stable (with basin of attraction the non-pure quadratic irrationals) and unstable (with repellor basin being the rational numbers). Since rationals and non-pure quadratic irrationals are intertwined in $[0,1]$, and indeed the cardinality of both sets is the same, periodic orbits are actually both stable and unstable with equal measure. \\

\noindent Also, we can also classify subfamilies of quadratic irrationals by the length of their orbit $m$.  Based on this property, it is natural to begin by considering the simplest class, i.e., pure quadratic irrationals whose continued fractions do not have a transient and have the shortest period of their periodic part, i.e. $n = 1$. These quadratic irrationals are solutions of the following equation
$$x^2 + bx-1=0, \; \; b\geq1$$
i.e. this is a subfamily of the wider class of quadratic polynomials where we fix $a=1, c=-1$, whose positive root $\phi_{b}^{-1}:=\frac{-b+\sqrt{b^2+4}}{2}$ is the fractional part of the {\it generalised Golden ratios}, or {\it metallic ratios} ($\phi_{1}^{-1}$ is the fractional part of the Golden ratio, whereas $\phi_{2}^{-1}$ ,the fractional part of $\sqrt{2}$, is the so-called \textit{Silver ratio}). The continued fraction expansion for this family adopts a simple shape $\phi_{b}^{-1}=[b,b,b, \dots ] := [\overline{b}]$. This implies that $\forall b\geq1,$ the quadratic irrational $\phi_{b}^{-1}$  forms periodic RG orbits of length $2b$, i.e.,  $T^{(2b)}(\phi_{b}^{-1})=\phi_{b}^{-1}$. The sequence of `coupling constants' $x$ along the periodic orbit is
$$\left\lbrace \phi_{b}^{-1}, T_1(\phi_{b}^{-1}), ,\dots,T_1^{(b-1)}(\phi_{b}^{-1}), T_2\circ T_1^{(b-1)}(\phi_{b}^{-1}), \dots, T_2^{(b-1)}\circ T_1^{(b-1)}(\phi_{b}^{-1}),T_1\circ T_2^{(b-1)}\circ T_1^{(b-1)}(\phi_{b}^{-1}) \right\rbrace ,$$
i.e. $(b-1)$ successive applications of $T_1$, then $(b-1)$ successive applications of $T_2$, and a final application of $T_1$ (the special case $b = 1$, the reciprocal of golden number, which orbit is $\left\lbrace \phi^{-1}, T_2(\phi^{-1}) = 1 - \phi^{-1} \right\rbrace$ . The precise values of this periodic orbit of period $2b$ is given by the following sequence of quadratic irrationals:
$$\phi_{b}^{-1},\bigg\lbrace \frac{\phi_{b}^{-1}}{1 - j\phi_{b}^{-1}} \bigg\rbrace_{j=1}^{j = b-1}, 1-\phi_{b}^{-1},\bigg\lbrace \frac{1 - (j+1)\phi_{b}^{-1}}{1 - j\phi_{b}^{-1}}\bigg\rbrace_{j=1}^{j =b-1}.$$

For this subfamily of pure quadratic irrationals, the associated family of non-pure quadratic irrationals is the so so-called \textit{generalized metallic ratios}, whose continued fraction defined as $[c_1, ... , c_r, \overline{b}]$. Starting from any generalized metallic ratio, the action of $\cal R$ will form a trajectory which ends in the orbit of the associated metallic ratio $\phi_{b}^{-1}$, i.e such periodic orbit is also a stable attractor for all the associated generalized metallic ratios, as previously discussed.

\noindent In conclusion, all quadratic irrational either belong to an unstable/stable periodic orbits, or form a trajectories converging onto them. 

\begin{figure}[htb]
\begin{center}
\includegraphics[width=1.1\textwidth, trim=2cm 0cm 0cm 0cm,clip=true]{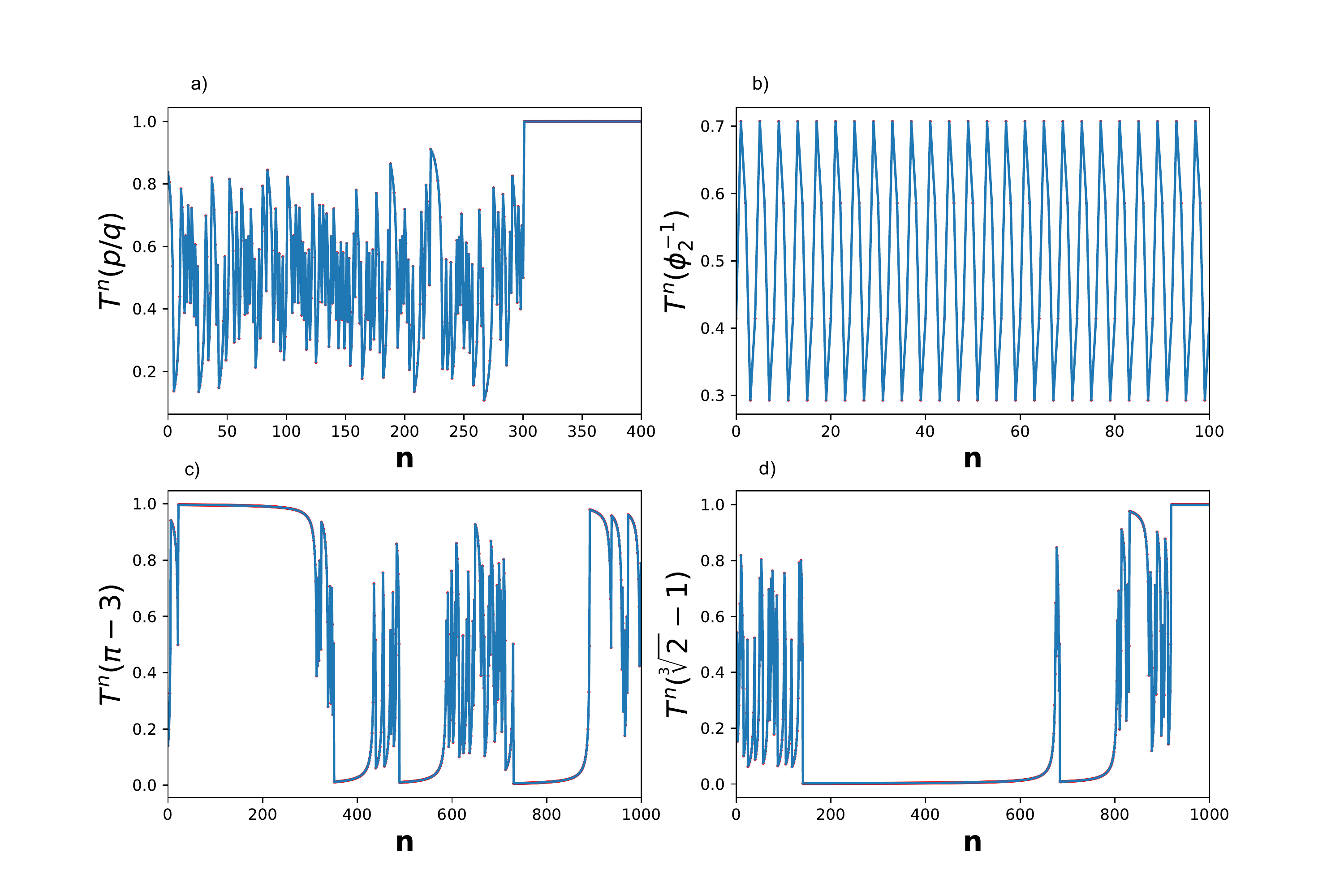}
\caption{Numeric simulations of algebraic operator $T$ for different real numbers using a fixed-point arithmetic with a precision of $1000$ significant places. \textbf{a)} Flow iterations for a rational number $p/q$ chosen by a random walk in the Farey binary tree. After a finite number of applications of $T$, the flow converges to $1/1$. \textbf{b)}  The flow derivated of the reciprocal of silver ratio $\phi_{2}^{-1} = \sqrt{2} - 1$ generates a periodic orbit of length $4$. \textbf{c)} and \textbf{d)} Taking as initial condition a trascendental number as the fractional part of $\pi$ or an algebraic number as $\sqrt[3]{2} - 1$, the flows are chaotic. In fact, $\pi - 3$ has been considered using the rational convergent of length $96$, whereas the $\sqrt[3]{2} - 1$ has been approximated by the rational convergent of length $65$. }
\label{Fig_X}
\end{center}
\end{figure}

\item {\bf Chaos}: All other initial conditions (e.g. non-quadratic irrational numbers) necessarily form aperiodic orbits (chaos). The cardinality of this set of initial conditions is indeed larger than the previous ones (rationals and quadratic irrationals). Indeed, by virtue of the topological conjugation between $T(x)$ and the shift map, $\cal R$ shows sensitive dependence on all but a subset of initial conditions of null Lebesgue measure --that subset being formed by the sets of all rationals and quadratic irrationals (both of these being countable)--. Let us investigate the structure of initial conditions that lead to chaotic behavior further.

\noindent First, consider {\bf transcendental} numbers (i.e., non algebraic irrationals) in $[0,1]$. Interestingly, if $x$ is a transcendental number, then the outcome of any algebraic function applied to $x$ is again transcendental. As the M\"obius transformation is indeed algebraic, then one can see that, besides the trivial case when the transformation is constant, initial conditions in the set of transcendental numbers in $[0,1]$ generate aperiodic orbits whose elements are all transcendental: we denote these orbits {\it transcendental  aperiodic orbits} (TAO) . Note however that, since the set of transcendental numbers is uncountable, there is in general no guarantee that any two transcendental numbers in $[0,1]$ will belong to the same TAO. By the numerability argument, it is easy to see that there actually exist an uncountable amount of different TAOs. Indeed, let us consider an arbitrary small neighbourhood of a transcendental number $x$, there are an uncountable amount of transcendental numbers, but the chaotic orbit of $x$, since it is countable, will only have at most a countable number of elements belonging to the neighbourhood of $x$. Thus, there are exist an uncountable number of points in that neighbourhood that do not belong to that orbit. 

\noindent Consider now the rest of {non-quadratic } irrational numbers, i.e. the set of {\bf non-quadratic algebraic} irrationals. Since the set of polynomials with rational constants is countable, trivially the set of non-quadratic algebraic irrationals  is also countable. Now, it is easy to prove that if $x$ is algebraic, then for a M\"obius transformation $M(x)$ (eq.\ref{eq:Mobius}) is also algebraic (a proof of this statement goes along the following line: if $x$ is algebraic then it is the root of an algebraic equation $P(x)=0$, for $P \in \mathbb{Q}[x]$. For $M(x)$ to be also algebraic, there needs to exist an algebraic equation $Q \in \mathbb{Q}[x]$ such that $Q(M(x))=0$. Trivially, $Q=P\circ M^{-1}$ (note that $M^{-1}$ always exist and is algebraic, and the composition of two algebraic functions is algebraic)). Moreover, the M\"obius transformation $M$ preserves the degree of algebraic number $x$. Indeed, let us suppose that $x$ is an algebraic number of degree $n$, i.e., $P(x) = a_n x^n + ... + a_1 x + a_0 = 0$, with $a_n \neq 0$, and let us define $Q(x)=(P\circ M^{-1})(x)$. Hence, $Q(M(x)) = (P\circ M^{-1})(M(x)) = P(x) = 0$. Then, as  
$$M^{-1}(x) = \frac{Dx -B}{-Cx + A},$$
the equation $Q(M(x)) = 0$ can be written as: 
\begin{eqnarray}
Q(M(x)) &=& a_n \left( \frac{Dx -B}{-Cx + A} \right)^n  + ... + a_1 \left( \frac{Dx -B}{-Cx + A} \right) + a_0 = 0 \Leftrightarrow \nonumber
\end{eqnarray}
\begin{eqnarray}
a_n(Dx -B)^n + ... + a_1(Dx -B)(-Cx + A)^{n-1} + a_0(-Cx + A)^n = 0 \nonumber
\end{eqnarray}
i.e., $Q$ is a polynomial of degree $n$. \noindent Accordingly, any initial condition which is a non-quadratic algebraic irrational of degree $n > 2$ will generate an aperiodic orbit whose elements are all non-quadratic algebraic irrationals with the same degree $n$: we denote these {\it algebraic aperiodic orbits} (AAOs). It is unclear whether a single AAO of degree $n$ is enough to recover the whole set of non-quadratic algebraic irrationals of degree $n$ in $[0,1]$ under iteration of $T$, or if there exist a countable union of different AAOs of degree $n$.\\
Be that as it may, we conclude that non-quadratic algebraic irrationals and transcendental numbers don't mix under the dynamics of $T$, so effectively $T$ splits out these two big families apart. By the topological conjugacy, the same is true for these families of Haros graphs: the RG flow $\cal R$ generates chaotic orbits of types TAO and AAO along the lines described above.
\end{enumerate}

For illustration, in Fig.\ref{Fig_X} we plot the iterations of $T$ for different initial conditions: a rational number (a), and numerical approximations of a pure quadratic irrational (b), a trascendental number (c) and a non-quadratic algebraic irrational (d). By topological conjugacy, these orbits are equivalent in the space of Haros graphs under the action of the RG operator $\cal R$.


\section{Entropy density towards and inside the invariant set of the RG flow}
\label{sec:ent}
One of the important questions stemming from the the structure of the RG flow is to explore how the variable coarse-graining induces an irreversible gradient along the RG action that could be observed in a monotonic function. This endeavour was crystallized in Quantum Field Theory (QFT) in the celebrated $c$-theorem for 2-dimensional QFTs \cite{zamolodchikov1986irreversibility} and its higher-dimensional version $a$-theorem \cite{cardy1988there}, where a specific function was shown to be monotonically decreasing along the RG flow, reaching a constant value at the fixed points. From a thermodynamic viewpoint, some authors have proposed a modified version of Boltzmann's $H$-theorem in the context of QFT \cite{gaite1996field}. Similarly, in the realm of random walk statistics \cite{Robledo} and dissipative chaotic dynamics \cite{QuasiperiodicGraphs,intermit,AnalyticalFeigenbaum, diaz2021logistic}, entropic gradients along the RG flow and extremization at the invariant set of RG flows has also been explored.\\
Much lesser analysis exists for the possible entropy change {\it inside} the invariant set of RG flows (attractor/repellor), if only because finding sufficiently complex RG flows with invariant sets beyond simple fixed points is in general difficult. A na\"ive and handwaving reading of the c-theorem, for instance, would suggest that at (and inside) the invariant set, entropy should remain constant. This principle is trivially fulfilled for RG flows with only a collection of stable and unstable fixed points, and checking such principle in more complex invariant sets has rarely been explored.\\
Our system, with a rich RG invariant set structure, allows to explore entropy gradients towards, away from, and inside invariant sets. In this section we therefore explore specific properties of the system as it evolves via the action of the RG operator $\cal R$. We consider the degree distribution $P(k;x)$ of the system --i.e. of the Haros graph $G_x$--, which computes
the probability that a node chosen randomly has degree $k$ in the Haros graph $G_x$. The degree distribution is the marginal distribution of the graph's degree sequence, and previous results suggest that such is a maximally informative property of the system \cite{canonical}. Interestingly, recent results allow us to have closed-form expressions for $P(k;x)$ when $x$ is rational \cite{DegreeHaros} or a generalized metallic ratio.  From the degree distribution, we subsequently compute its entropy
$$S(x) := - \sum_{k\geq 2} P(k;x)\cdot \log(P(k;x)).$$
Interpreting the degree $k$ as a random variable, $S(x)$ is the Shannon entropy of its distribution. 

The first three values $k =2,3,4$ of $P(k;x)$ are known for every real number, having $P(2;x) = x$, $P(3,x) = 1 - 2x$, $P(4;x) = 0$ if $x < 1/2$; and $P(2;x) = 1-x$, $P(3;x) = 2x - 1$, $P(4;x) = 0$ if $x > 1/2$. 
In left panel of Fig. \ref{Fig_5}, $S(x)$ has been depicted for all Haros graph $G_x$, where $x \in \mathcal{F}_{1000}$. A visual exploration of the figure discloses how the Haros graphs $G_{1/n}$ are local minima of $S(x)$. Furthermore, it has been demonstrated that the local maxima of $S(x)$ are precisely the set of quadratic irrationals (see \cite{HarosPaper} for proofs and an in-depth discussion of the properties of $S(x)$ for different families of numbers). The proof uses the Maxent technique  \cite{Jayne}, where the Lagrangian functional to optimize is 
\begin{equation*}
\mathcal{L}[\left\lbrace P(k,x) \right\rbrace]=-\sum_{k=5}^{\infty}{P(k,x)\log{P(k,x)}}-(\lambda_0-1)\left(\sum_{k=5}^{\infty}{P(k,x)}-{\cal Q}_0 \right)
-\lambda_1\left(\sum_{k=5}^{\infty}{kP(k,x)}- {\cal Q}_1 \right),
\end{equation*}
being the constrains (for $x < 1/2$), the normalisation of $P(k,x)$: $${\cal Q}_0 = \sum_{k=5}^{\infty} P(k,x) = x$$ and the constant value for arithmetic mean of degree sequence $\langle k \rangle = 4$: $${\cal Q}_1 = \sum_{k=5}^{\infty} k\cdot P(k,x) = 4x + 1.$$
Setting other values $P(k,x) = 0$ as well, the optimum value is associated to different quadratic irrationals.\\  

Let us study how the degree distribution changes along the RG periodic orbit. As we previously mentioned, the metallic ratios $\phi_{b}^{-1} = [\overline{b}]$, with $b \geq 1$, form periodic RG orbits of length $2b$. 
First, metallic ratios have the following degree distribution for $b\geq 2$:
\begin{equation}
P(k,\phi_{b}^{-1})=\left\{
\begin{array}{ll}
\phi_{b}^{-1}  & , k=2 \\
1-2\cdot \phi_{b}^{-1}& , k=3 \\
(1 - \phi_{b}^{-1})\cdot (\phi_{b}^{-1})^{\frac{k - 3}{b}}   & , k=bn + 3, n \in \mathbb{N} \\
0 &  , \text{otherwise.}%
\end{array}%
\right.
\label{Metb_eq}
\end{equation} 
Second, to unveil the degree distribution of the Haros graphs belonging to the orbit $T^{j}(\phi_{b}^{-1})$, we first need point out that it is easy obtain the following scaling equation for degrees $k \geq 5$:
\begin{equation}
P(k+1,x) =\left\{
\begin{array}{ll}
(1 - x)\cdot P\left(k, T_1(x) \right) & , \textit{ if } x \leq 1/2, \\
x \cdot P\left(k, T_2(x) \right) & , \textit{ if } x \geq 1/2.
\end{array}%
\right.
\label{Pscaling1}
\end{equation}
To verify the validity of Eq. \ref{Pscaling1} above, we show first that $T_{1}^{-1}(x) = \frac{x}{x+1}$ if $x < 1/2$ and $T_{2}^{-1}(x) = \frac{1}{2-x}$ if $x > 1/2$. Secondly, we use a result from \cite{HarosPaper} (see Eq. 6 in \cite{HarosPaper}) where, $\forall x \in [0,1]$:
\begin{equation}
P(k,x) =\left\{
\begin{array}{l}
(1 + x)\cdot P\left(k+1, \frac{x}{1 + x}\right)  \\
(2-x)\cdot P\left(k+1, \frac{1}{2-x}\right).
\end{array}%
\right.
\label{Pscaling}
\end{equation}
Combining Eq. \ref{Metb_eq} and Eq. \ref{Pscaling1},  and with a bit of algebra, the degree distribution of $P(k, T_{1}^{(j)}(\phi_{b}^{-1}))$ for values $j = 0, ... , b-2$ results:
\begin{equation}
P(k,T_{1}^{(j)}(\phi_{b}^{-1}))=\left\{
\begin{array}{ll}
T_{1}^{(j)}(\phi_{b}^{-1}) & , k=2 \\
1-2\cdot T_{1}^{(j)}(\phi_{b}^{-1}) & , k=3 \\
\frac{1 - \phi_{b}^{-1}}{1 - j \cdot \phi_{b}^{-1}}\cdot (\phi_{b}^{-1})^{\frac{k - 3 + j}{b}}   & , k=bn + 3 - j, n \in \mathbb{N} \\
0 &  , \text{otherwise.}%
\end{array}%
\right.
\label{TMetb_eq}
\end{equation} 
It is easy to see that $T_{1}^{(b-1)}(\phi_{b}^{-1}) > 1/2$. Hence, the Haros graph associated to this number has a degree distribution:

\begin{equation}
P(k,T_{1}^{(b-1)}(\phi_{b}^{-1}))=\left\{
\begin{array}{ll}
1 -T_{1}^{(b-1)}(\phi_{b}^{-1}) & , k=2 \\
2\cdot T_{1}^{(b-1)}(\phi_{b}^{-1}) - 1& , k=3 \\
\frac{1 - \phi_{b}^{-1}}{1 - (b-1) \cdot \phi_{b}^{-1}}\cdot (\phi_{b}^{-1})^{\frac{k - 4 + b}{b}}   & , k=bn + 4, n \in \mathbb{N} \\
0 &  , \text{otherwise}.
\end{array}%
\right.
\label{Tb1Metb_eq}
\end{equation} 
The following element of the orbit is $T_2 \circ T_{1}^{(b-1)}(\phi_{b}^{-1}) = 1 - \phi_{b}^{-1}$, and by symmetry, we have $P(k,x) = P(k, 1-x)$. Therefore, the remaining $b$ elements of the orbit can be calculated via the second scaling of Eq. \ref{Pscaling1}, or using the symmetry. This concludes the study of the degree distribution of Haros graphs along the periodic orbit.  \\


Let us then consider the graph entropy $S(x)$ illustrated in Fig. \ref{Fig_5}, with a clear  self-similar structure in the subintervals $[1/n, 1/(n+1)]$ for $n \geq 2$. In order to assign null entropy to Haros graph $G_{1/n}$, we first define a reduced entropy $H(x)$ as:
\begin{equation}
H(x) \equiv\left\{
\begin{array}{ll}
S(x) + 2x\cdot \log(x) + (1-2x)\cdot \log(1-2x)  & , \; x \in [0,1/2] \\
S(x) + 2(1-x)\cdot \log(1-x) + (2x-1)\cdot \log(2x-1)  & , \; x \in [1/2,1].
\label{def:EntH}
\end{array}%
\right.
\end{equation}
If we denote:
\begin{equation}
\Lambda(x) = \frac{x\log x -  (1-2x)\log(1-2x)  -  (3x-1)\log(3x-1)}{x},
\label{Lambda}
\end{equation}
then we can prove that the reduced entropy verifies the following scaling equations (a proof can be found in Appendix \ref{sec:AppendixA}):
\begin{equation}
H(x)=\left\{
\begin{array}{ll}
(1-x)\cdot H(T_1(x)) & , \; x \in [0,1/3] \\
(1-x)\cdot H(T_1(x)) + x\cdot \Lambda(x) & , \; x \in [1/3,1/2] \\
x\cdot H(T_2(x)) + (1-x)\cdot \Lambda(1-x) & , \; x \in [1/2,2/3] \\
x\cdot H(T_2(x)) & , \; x \in [2/3,1].
\end{array}%
\right.
\label{Hscaling1}
\end{equation}
Moreover, we now define the entropy density as:
\begin{equation}
h(x)=\left\{
\begin{array}{ll}
\frac{H(x)}{x} & , \; x \in [0,1/2] \\
\frac{H(x)}{1-x} & , \; x \in [1/2,1].
\end{array}%
\right.
\label{Hdensity}
\end{equation}

\begin{figure}[h]
\begin{center}
\includegraphics[width=1.1\textwidth, trim=2cm 0cm 0cm 0cm,clip=true]{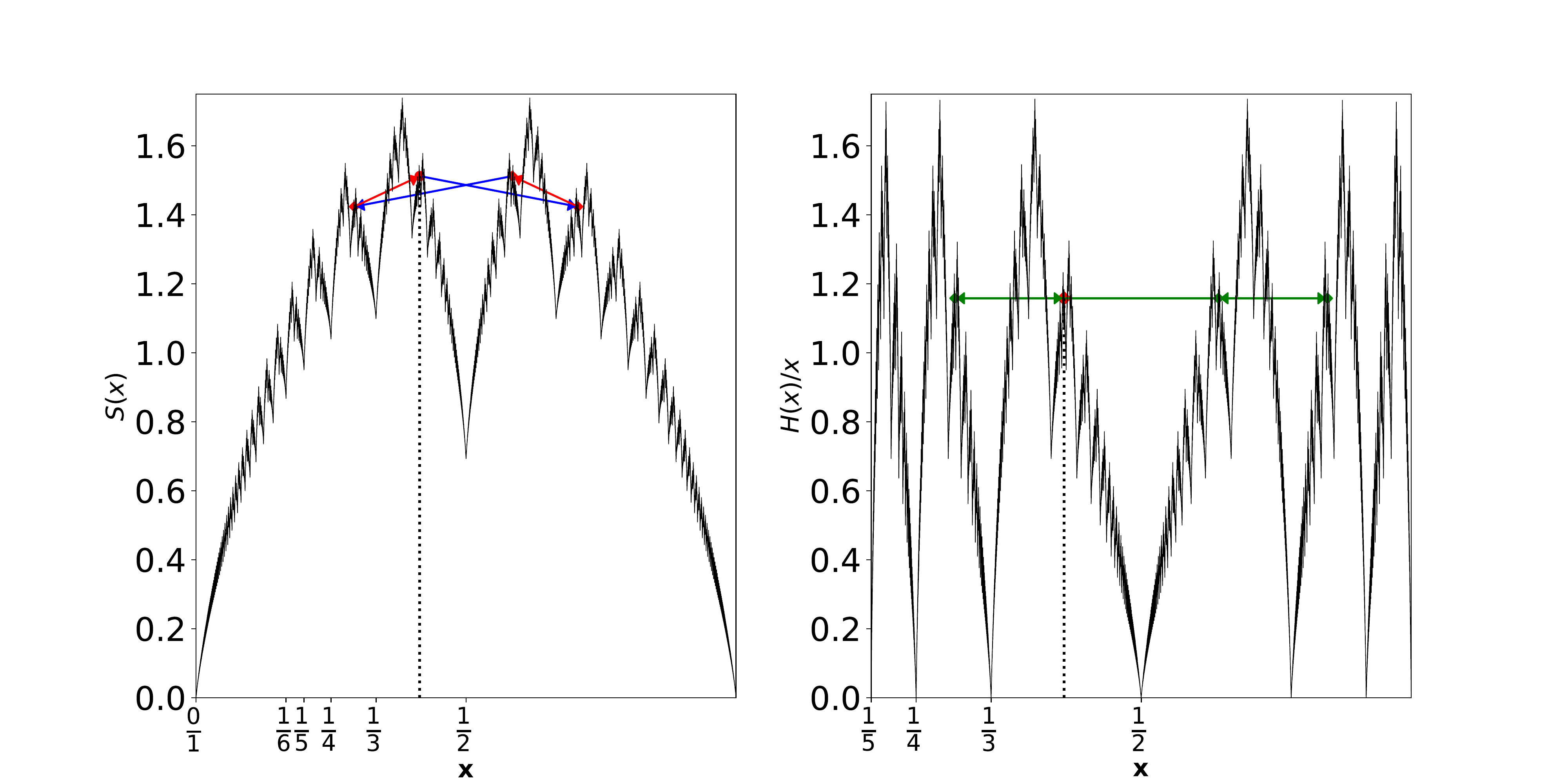}
\caption{Entropy graph function $S(x)$ (left panel) and reduced density entropy graph function $h(x) = H(x)/x$ defined in Eq. \ref{Hdensity} (right panel), computed for all Haros graphs $G_x$ with $x \in {\cal F}_{1000}$. The silver ratio $\phi_2 = [\overline{2}]$ (pointed with dotted vertical line) generates a periodic orbit of length $m = 4$. The entropy of  the elements of the orbit fluctuates in $S(x)$, whereas is constant in $h(x)$.}
\label{Fig_5}
\end{center}
\end{figure}

Hence, applying this definition to \ref{Hscaling1},we obtain the scaling equations:
\begin{equation}
h(T_1(x))=\left\{
\begin{array}{ll}
h(x) & , \; x \in [0,1/3] \\
\frac{x}{1-2x} \cdot(h(x) - \Lambda(x))&, \; x \in [1/3,1/2]
\end{array}%
\right.
\label{HDscaling1}
\end{equation}
\begin{equation}
h(T_2(x))=\left\{
\begin{array}{ll}
\frac{1-x}{2x-1} \cdot(h(x) - \Lambda(1-x)) & , \; x \in [1/3,1/2] \\
h(x) & , \; x \in [2/3,1]
\end{array}%
\right.
\label{HDscaling2}
\end{equation}



\noindent Equipped with these equations, we can now outline several important consequences.\\
First, if $x$ and $T(x)$ are both in $[0, 1/2]$ or $[1/2, 1]$, the renormalization operator $\mathcal{R}$ preserves the entropy density $h$. This means that if the renormalization trajectory for an initial rational value remains in one or the other half side of the interval $[0,1]$, such flow has constant entropy density. This situation, however, only happens for the rationals $p/q=1/n$ and $p/q=(n-1)/n$, for other rationals, the trajectory hops, and therefore the entropy density flow fluctuates, and Eq. \ref{HDscaling1} and Eq. \ref{HDscaling2} allows us control such fluctuation. In panel (a) of Fig.\ref{fig:T_h_example} we illustrate the iteration $T^{(n)}(x)$ (black dots) and the entropy density $h(T^{(n)}(x))$ (blue hollow points) for a specific example, where we can see that abovementioned signatures.\\
The question is then, what is the precise shape of such entropy density function along the RG flow, both towards/away from the RG invariant set and, interestingly, inside the RG set?\\

\begin{figure}[htb]
\begin{center}
\includegraphics[width=0.47\textwidth]{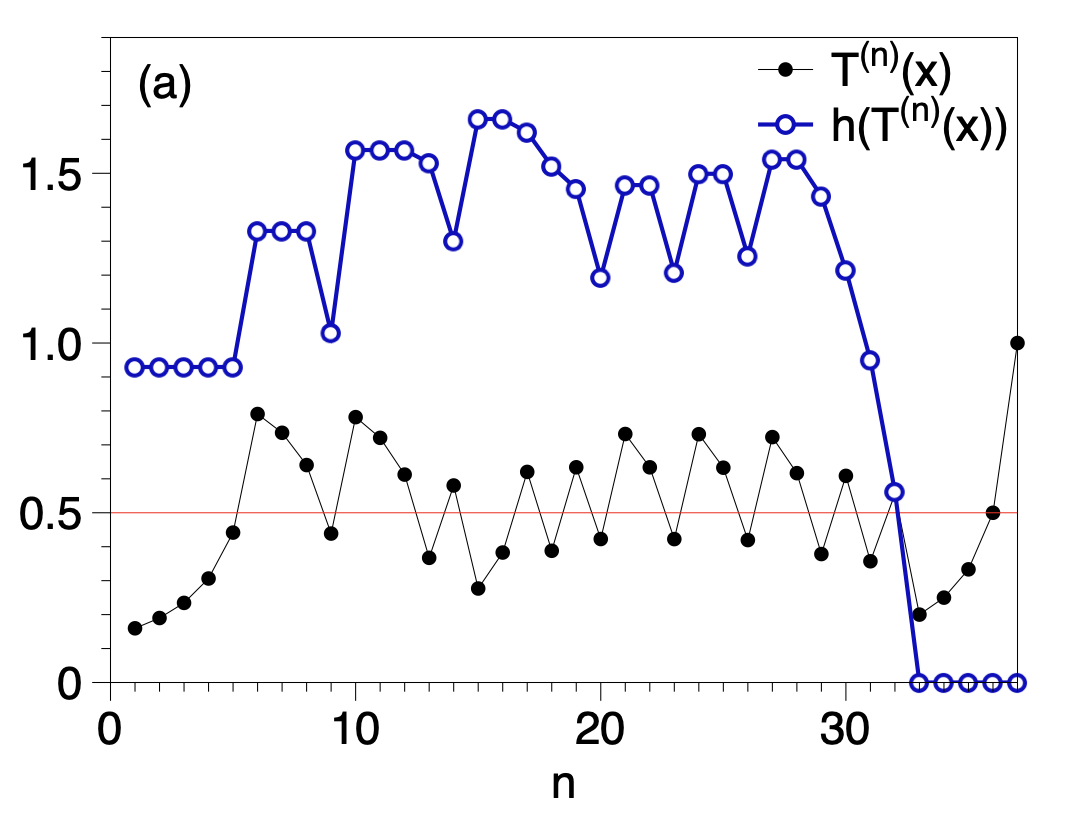}
\includegraphics[width=0.47\textwidth]{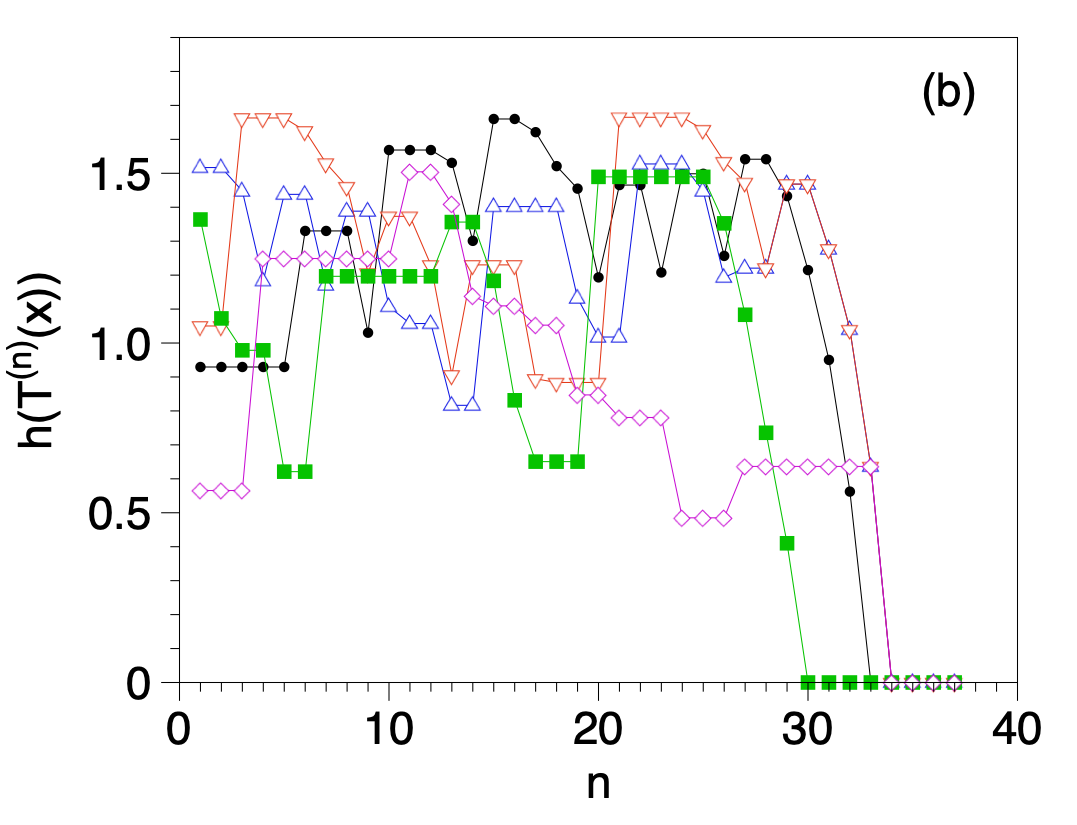}
\caption{(Panel a) $T^{(n)}(x)$ (black dots) and the entropy density $h(T^{(n)}(x))$ (blue hollow points) as a function of the iteration $n$, for a concrete example of a rational $x$. We can see that the entropy density remains constant when the iteration of $T$ leaves the point in the same side of the interval $[0,1]$, and changes when the side is changed. See the text for details. (Panel b) Entropy density $T^{(n)}(x)$ as a function of $n$ for five different rationals, generated randomly (see the text).}
\label{fig:T_h_example}
\end{center}
\end{figure}

\noindent {\bf Convergence to the stable fixed point. } Since rational Haros graphs renormalize into  $G_{1/1}$ --which has vanishing entropy density-- after a finite number of steps, $\forall x=p/q, \exists r_0 \in \mathbb{N}$ such that $\forall r \geq r_0$:
\begin{equation}
h(T^{(r)}(p/q)) = 0
\end{equation}
(it can be check that if $p/q = [a_1, ... , a_n]$, then $r_0$ is obtained as  $\sum_{i=1}^{n} a_i$). In other words, entropy density flows towards its minimum along the flow towards the stable fixed point of the RG flow. This could seem, in principle, to be in agreement with a $c$-like theorem in this context. Note, however, that such convergence is far from smooth. Much on the contrary, depending on the initial condition $G_{p/q}$, the flow towards $G_{1/1}$ will show fluctuations in the entropy density with local increments and decrements. An illustration of such behavior for different initial values $x=p/q$ is shown in panel (b) of Fig.\ref{fig:T_h_example}. \\

\noindent {\bf Entropy-density is constant inside (and only inside?) periodic orbits. } It can be proved (see Appendix \ref{sec:AppendixB}) that for the family of metallic ratios $\phi_{b}^{-1}$, the cumbersome entropy change showed in Eq. \ref{HDscaling1} and Eq. \ref{HDscaling2} is simplified, obtaining a remarkable entropy-constant result inside each periodic orbit:
\begin{equation}
  h(T^{(n)}(\phi_{b}^{-1}))= h(\phi_{b}^{-1}), \ \forall n
  \label{eq:constant}
\end{equation}
In the right panel of Fig \ref{Fig_5} we illustrate this behavior for the specific case of the silver ratio $\phi_{2}^{-1}$. We {conjecture} that such constancy is only possible inside these orbits, and for any other type of irrational, the entropy density necessarily changes along the RG flow (a rigorous proof for this observation remains elusive and is an open problem).\\


\noindent {\bf Entropy density along aperiodic orbits. } The behaviour of entropy density gradient along the RG flow of non-quadratic algebraic and trascendental numbers (for which the RG flow is aperiodic) has not been elucidated and remains an open problem. According to the conjecture above, we conjecture that such entropy density is, at least, non constant.

\section{Conclusions}
\label{sec:con}

In this paper we have thoroughly studied the iteration dynamics of the graph-theoretic operator $\cal R$ over the set of Haros graphs $\cal G$, a system which is in bijection with the reals $\mathbb{R}$. We have shown that $\cal R$ performs a renormalization of Haros graphs and that $\cal G$ is indeed renormalizable. We proved that the iterated dynamics of $\cal R$ over Haros graphs (which we call the RG flow) is chaotic, as its dynamics is topologically conjugate to an interval chaotic map which is itself conjugate to the shift map via Minkowski's question mark function. We explicitly found the conjugacy between $\cal R$ and the chaotic interval map and thus are able to analytically study the chaotic RG flow in the `coupling constant space', as it is routinely done in standard RG treatments of physical systems. Incidentally, the fact that the RG flow of Haros graphs is conjugate to the shift map suggests some possible connection to spin systems with complex coupling constants, which were recently  shown to possess a chaotic (Bernoulli) RG structure \cite{bosschaert2022chaotic}.\\
The RG invariant set consists of an extremely rich hierarchy of periodic orbits and non-mixing aperiodic orbits, and we showed that this hierarchy provide a dynamical classification of different families of the reals, including quadratic irrationals, non-quadratic algebraic irrationals and trascendental numbers.\\
We finally studied an entropic density function and how this function evolved along the RG flow, i.e. towards the stable fixed point and inside the invariant set. The stable fixed point of the RG flow has minimum entropy density, whereas the periodic orbits are local entropic maxima. These results provide a further verification of the apparent extremal entropy properties of RG attractors and repellors \cite{Robledo}. Since periodic orbits are unstable with respect to the  measure of rational numbers --i.e. the RG flow connects the periodic orbit and the stable fixed point for perturbations close to the orbit--, then the entropy density is necessarily {\it globally} decreasing. However, we haven't found that such function is locally monotonic: along its way towards the stable fixed point, the trajectory can show locally both increases and decreases of entropy density.  
Finally, we proved that entropy density is (nontrivially) constant inside periodic orbits of the RG. These results can be put in comparison with RG monotonicity theorems in quantum field theories, which are folklorically believed to preclude the existence of complex RG flows (see, however, \cite{curtright2012renormalization}). In that sense, we argue that observing constant entropy density inside the periodic orbit is in perfect agreement with monotonicity theorems, that actually require that the coupling constant flow has flat derivative at the invariant set.\\
Finally, we wish to stress that the system under study is purely formal, after all it is a graph-theoretic set $\cal G$ in bijection with $\mathbb{R}$ \cite{HarosPaper}. In this sense, it is difficult to provide a direct physical interpretation of its chaotic RG flow, where physical chaotic RG flows speak of an intrinsically multiscale system, with broken scale-invariance. To further shed light on such interpretation, it is paramount to investigate the possible statistical mechanic structure underlying $\cal G$, and this will be the topic of further research. Overall, this work further highlights the rich structure of Haros graphs and the unexpected links between number theory, graph theory and theoretical physics \footnote{\url{https://empslocal.ex.ac.uk/people/staff/mrwatkin/zeta/surprising.htm}}.  

\section*{Appendix}
\subsection{Proof of Equation \ref{Hscaling1}}
\label{sec:AppendixA}
\begin{proof}
If $x < 1/3$, thus it is clear that $T_1(x) < 1/2$ and $P(5,x) = 0$:
\begin{eqnarray*}
H(x) &  = & -\sum_{k \geq 5} P(k,x)\log(P(k,x)) + x\cdot \log x \\
& = & -\sum_{k \geq 5} (1-x)P(k-1,T_1(x))\log((1-x)P(k-1,T_1(x))) + x\cdot \log x  \\
& = & (1-x)\left(-\sum_{k \geq 4}P(k,T_1(x))\log(1-x) -\sum_{k \geq 4}P(k,T_1(x))\log(P(k,T_1(x))) + \frac{x}{1-x}\cdot \log x \right)\\
& = & (1-x)\left(-\frac{x}{1 - x}\log(1-x) -\sum_{k \geq 4}P(k,T_1(x))\log(P(k,T_1(x))) + \frac{x}{1-x}\cdot \log x \right) \\
& = & (1-x)\left(-\sum_{k\geq 5}P(k,T_1(x))\log(P(k,T_1(x)))  + \frac{x}{1 - x} \cdot \log \frac{x}{1 - x} \right)\\
& = & (1-x)\left(-\sum_{k \geq 5}P(k,T_1(x))\log(P(k,T_1(x)))  + T_1(x) \cdot \log T_1(x)  \right)\\
& =&  (1-x)\cdot H(T_1(x)).
\end{eqnarray*} 
If we take $1/3 < x < 1/2$, then $T_1(x) > 1/2$ and $P(5,x) = 3x - 1$. We have the following equation:
\begin{eqnarray*}
H(x) &  = & -\sum_{k \geq 5} P(k,x)\log(P(k,x)) + x\cdot \log x \\
& = & -\sum_{k \geq 6} (1-x)P(k-1,T_1(x))\log((1-x)P(k-1,T_1(x))) - P(5,x) \cdot \log P(5,x)  + x\cdot \log x  \\
& = & (1-x)\left(-\sum_{k \geq 5}P(k,T_1(x))\log(1-x) -\sum_{k \geq 5}P(k,T_1(x))\log(P(k,T_1(x)))  \right) \\
& + & x\cdot \log x - P(5,x) \cdot \log P(5,x)\\
& = & (1-x)\left(-(1-T_1(x)) \cdot \log(1-x) -\sum_{k \geq 5}P(k,T_1(x))\log(P(k,T_1(x))) \right) \\
& + & x\cdot \log x - P(5,x) \cdot \log P(5,x)\\
& = & (1-x)\left[ -(1-T_1(x)) \log(1-x) -\sum_{k \geq 5}P(k,T_1(x))\log(P(k,T_1(x))) + \right. \\
& & \Biggl. (1-T_1(x))\log(1-T_1(x)) -(1-T_1(x))\log(1-T_1(x))   \Biggr] + x\cdot \log x - P(5,x) \cdot \log P(5,x)\\
& = & (1-x)\cdot \Biggl[ -(1-T_1(x)) \log(1-x) + H(T_1(x)) -(1-T_1(x))\log(1-T_1(x)) \Biggr] \\
&+& x\cdot \log x - P(5,x) \cdot \log P(5,x)\\
&&\text{\dots Applying $1-T_1(x) = (1-2x)/(1-x)$} \dots\\
& = & (1-x)H(T_1(x)) -(1-2x) \log(1-x) -(1-2x) \log(1-2x) + (1-2x) \log(1-x) \\
&+& x\cdot \log x - P(5,x) \cdot \log P(5,x)\\
&=& (1-x)H(T_1(x)) +P(2,x) \cdot \log P(2,x) - P(3,x) \cdot \log P(3,x) - P(5,x) \cdot \log P(5,x)
\end{eqnarray*} 
Analogously, if $x > 2/3$, thus it is clear that $T_2(x) > 1/2$ and $P(5,x) = 0$:
\begin{eqnarray*}
H(x) &  = & -\sum_{k \geq 5} P(k,x)\log(P(k,x)) + (1-x)\cdot \log (1-x) \\
& = & -\sum_{k \geq 5} x\cdot P(k-1,T_2(x))\log(x\cdot P(k-1,T_2(x))) +  (1-x)\cdot \log (1-x)  \\
& = & x \cdot \left(-\sum_{k \geq 4}P(k,T_2(x))\log(x) -\sum_{k \geq 4}P(k,T_2(x))\log(P(k,T_2(x))) + \frac{1-x}{x}\cdot \log(1-x) \right)\\
& = & x \cdot\left(- \left( 1 - \frac{2x-1}{x}\right) \log(x) -\sum_{k \geq 4}P(k,T_2(x))\log(P(k,T_2(x))) + \frac{1-x}{x}\cdot \log(1-x) \right) \\
& = & x \cdot \left(-\sum_{k\geq 5}P(k,T_2(x))\log(P(k,T_2(x)))  + \frac{1-x}{x} \cdot \log \frac{1-x}{x} \right)\\
& = & x \cdot \left(-\sum_{k \geq 5}P(k,T_2(x))\log(P(k,T_2(x)))  + T_2(x) \cdot \log T_2(x)  \right)\\
& =&  x \cdot H(T_2(x)).
\end{eqnarray*} 
If we take $1/2 < x < 2/3$, then $T_2(x) < 1/2$ and $P(5,x) = 2 - 3x$. We conclude the proof with the following equation:
\begin{eqnarray*}
H(x) &  = & -\sum_{k \geq 5} P(k,x)\log(P(k,x)) + (1-x)\cdot \log (1-x) \\
& = & -\sum_{k \geq 6} x \cdot P(k-1,T_2(x))\log(x \cdot P(k-1,T_2(x))) - P(5,x) \cdot \log P(5,x)  + (1-x)\cdot \log (1-x)  \\
& = & x \cdot \left(-\sum_{k \geq 5}P(k,T_2(x))\log(x) -\sum_{k \geq 5}P(k,T_2(x))\log(P(k,T_2(x)))  \right) \\
& + & (1-x)\cdot \log (1-x) - P(5,x) \cdot \log P(5,x)\\
& = & x \cdot \left(-T_2(x) \cdot \log(x) -\sum_{k \geq 5}P(k,T_2(x))\log(P(k,T_2(x))) \right) \\
& + & (1-x)\cdot \log (1-x) - P(5,x) \cdot \log P(5,x) \\
& = & x \cdot \left[ -T_2(x) \cdot \log(x) -\sum_{k \geq 5}P(k,T_2(x))\log(P(k,T_2(x))) + \right. \\
& & \Biggl. T_2(x) \cdot \log(T_2(x)) -T_2(x) \cdot \log(T_2(x))   \Biggr] + (1-x)\cdot \log (1-x) - P(5,x) \cdot \log P(5,x)\\
& = & x \cdot \Biggl[ -T_2(x) \cdot \log(x) + H(T_2(x)) -T_2(x) \cdot \log(T_2(x))  \Biggr] \\
&+& (1-x) \cdot \log (1-x) - P(5,x) \cdot \log P(5,x)\\
& = & x \cdot H(T_2(x)) -(2x-1) \log(x) -(2x-1) \log(2x-1) + (2x-1) \log(x) \\
&+& (1-x) \cdot \log (1-x) - P(5,x) \cdot \log P(5,x)\\
&=& x\cdot H(T_2(x)) +P(2,x) \cdot \log P(2,x) - P(3,x) \cdot \log P(3,x) - P(5,x) \cdot \log P(5,x)
\end{eqnarray*} 
\end{proof}

\subsection{Proof of constant entropy density in the metallic ratio orbits}
\label{sec:AppendixB}
{\it Proof.} Considering the metallic ratios $\phi_{b}^{-1}= [0;\overline{b} ]$, it is clear that $T_{1}^{(i)}(\phi_{b}^{-1}) < 1/2$ for $i = 1, ... , b-2$ and $T_{1}^{(b-1)}(\phi_{b}^{-1}) > 1/2$. Moreover, as  $\phi_{b}^{-1}$ is the solution of the quadratic equation $x^2 +bx - 1 = 0$, then $\phi_{b}^{-1}$ verifies the following relationships:
\begin{equation}
    1 - b\phi_{b}^{-1} =\phi_{b}^{-2}
\label{met1}
\end{equation}
\begin{equation}
    (b+1)\phi_{b}^{-1}- 1 = \phi_{b}^{-1} - \phi_{b}^{-2}.
\label{met2}
\end{equation}
Hence, if we define $\Lambda(x)$ (see Eq. \ref{Lambda}) and $\beta := T_{1}^{(b-2)}(\phi_{b}^{-1}) = \frac{\phi_{b}^{-1}}{1-(b-2)\phi_{b}^{-1}}$, we have:
\begin{eqnarray*}
\Lambda(\beta)  &  = & \Biggl[  \frac{\phi_{b}^{-1}}{1-(b-2)\phi_{b}^{-1}} \cdot \log \left( \frac{\phi_{b}^{-1}}{1-(b-2)\phi_{b}^{-1}} \right) -  \frac{1 -b\phi_{b}^{-1}}{1-(b-2)\phi_{b}^{-1}} \cdot \log \left( \frac{1-b\phi_{b}^{-1}}{1-(b-2)\phi_{b}^{-1}} \right) - \Biggr.  \\
& - & \Biggl. \frac{(b+1)\phi_{b}^{-1}-1}{1-(b-2)\phi_{b}^{-1}} \cdot \log \left( \frac{(b+1)\phi_{b}^{-1}-1}{1-(b-2)\phi_{b}^{-1}} \right) \Biggr] \cdot \frac{1-(b-2)\phi_{b}^{-1}}{\phi_{b}^{-1}} \\
&&\textit{(Using Eq. \ref{met1} and \ref{met2})} \\
& = & \log(\phi_{b}^{-1}) - \log(1 -(b-2)\phi_{b}^{-1}) - \phi_{b}^{-1}\log(\phi_{b}^{-2}) + \phi_{b}^{-1} \log(1-(b-2)\phi_{b}^{-1}) \\
& - & (1-\phi_{b}^{-1})\log(\phi_{b}^{-1}\cdot(1-\phi_{b}^{-1})) + (1-\phi_{b}^{-1})\cdot \log(1 -(b-2)\phi_{b}^{-1}) \\
& = & \log(\phi_{b}^{-1}) - 2\phi_{b}^{-1} \log(\phi_{b}^{-1}) -(1-\phi_{b}^{-1})\log(\phi_{b}^{-1}) -(1-\phi_{b}^{-1})\log(1-\phi_{b}^{-1}) \\
& = & -\phi_{b}^{-1}\log(\phi_{b}^{-1}) -(1-\phi_{b}^{-1})\log(1-\phi_{b}^{-1})
\end{eqnarray*}  
Thus, applying the equation \ref{Hscaling1} $b-1$ times over the metallic ratio $\phi_{b}^{-1}$, we obtain:
$$h(\beta) = h(\phi_{b}^{-1}) = - \log \left(1 - \phi_{b}^{-1} \right) - \frac{1+ \phi_{b}^{-1}}{b}\cdot \log (\phi_{b}^{-1}),$$
and, therefore, we can establish that:
\begin{eqnarray*}
h(\beta) &= &\frac{1-T_1(\beta)}{T_1(\beta)}\cdot h(T_1(\beta)) + \frac{M(\beta)}{\beta}  =    \phi_{b}^{-1} \cdot h(T_1(\beta)) - \phi_{b}^{-1}\log(\phi_{b}^{-1}) -(1-\phi_{b}^{-1})\log(1-\phi_{b}^{-1}) \\
 & = &  - \log \left(1 - \phi_{b}^{-1} \right) - \frac{1+ \phi_{b}^{-1}}{b}\cdot \log (\phi_{b}^{-1}) \Longrightarrow \\
h(T_1(\beta)) & = & -\log(1-\phi_{b}^{-1})\cdot \left( \frac{1-1+\phi_{b}^{-1}}{\phi_{b}^{-1}} \right) - \log(\phi_{b}^{-1})\cdot \left(\frac{1 + \phi_{b}^{-1}}{b\phi_{b}^{-1}} -1 \right) \\
& = & - \log \left(1 - \phi_{b}^{-1} \right) - \frac{1+ \phi_{b}^{-1}}{b}\cdot \log (\phi_{b}^{-1}) = h(\beta),
\end{eqnarray*}
concluding that the metallic ratio orbits preserve the density entropy $h$.    \qedsymbol\\

\noindent {\bf Acknowledgments}  JC and BL acknowledge funding from Spanish Ministry of Science and Innovation under project M2505 (PID2020-113737GB-I00). LL acknowledges funding from project DYNDEEP (EUR2021-122007) and MISLAND (PID2020-114324GB-C22), both projects funded by the Spanish Ministry of Science and Innovation, and Severo Ochoa and Mar\'ia de Maeztu Program for Centers and Units of Excellence in R\&D (MDM-2017-0711) funded by MCIN/AEI/10.13039/501100011033.  \\

\bibliographystyle{ieeetr}
\bibliography{Renormalization_A_dynamic_over_Farey_Graphs.bib}

\end{document}